\documentclass[11pt]{amsart}
\usepackage{amssymb,epsfig}

\allowdisplaybreaks[4]

\newtheorem{thm}{Theorem}[section]
\newtheorem{conj}[thm]{Conjecture}%[section]
\newtheorem{lem}[thm]{Lemma}%[section]
\newtheorem{rem}[thm]{Remark}%[section]
\newtheorem{prop}[thm]{Proposition}%[section]
\newtheorem{defi}[thm]{Definition}%[section]

\def\smirnovcras{0985.60090}
\def\kestenwalks{0626.60067}
\def\cardyriva{MR2280251}
\def\verblunsky{0034.36303}
\def\lswust{MR2044671}
\def\ferrandbook{MR0069895}

\def\baxterbook{MR998375}

\def\mccoywubook{mccoywu-book}

\def\grimmettbookrc{grimmett-bookrc}

\def\yang{MR0051740}

\def\kaufmanonsageriv{kaufmanonsageriv}

\def\lswust{MR2044671}

\def\mccreawhipple{MR0002733}

\def\courantetal{54.0486.01}

\def\sle#1{$\mathrm{SLE}\left(#1\right)$}
\def\slee{$\mathrm{SLE}$}
\def\br#1{\left(#1\right)}
\def\Bbr#1{\Big(#1\Big)}

\def\brs#1{\left\{#1\right\}}
\def\abs#1{\left|#1\right|}
\def\norm#1{\left\|#1\right\|}

\def\dist#1{{{\mathrm{dist}}\br{#1}}}
\def\const{{\mathrm{const}}}

\newcommand{\DD}{{\mathbb D}}
\newcommand{\ZZ}{{\mathbb Z}}
\newcommand{\CC}{{\mathbb C}}

\newcommand{\EE}{{\mathbb E}}
\newcommand{\PP}{{\mathbb P}}
\newcommand{\RR}{{\mathbb R}}
\newcommand{\IM}{{\mathrm{Im}}}
\newcommand{\RE}{{\mathrm{Re}}}

\newcommand{\oo}{{o}}
\def\Exp#1{{\EE\left(#1\right)}}
\def\Exps#1#2{{\EE_{#1}\left(#2\right)}}
\def\Prob#1{{\PP\left(#1\right)}}
\def\Proj#1{{\mathrm{Proj}}\left(#1\right)}

\def\HH{{H}}
\def\spin{\sigma}
\def\Om{\Omega}
\def\ga{\gamma}
\def\si{\lambda}
\def\wi{{w}}
\def\num{{n}}
\def\we{{W}}
\def\mesh{\delta}

\def\cara{\stackrel{{Cara}}{\longrightarrow}}
\newcommand\unif{{\rightrightarrows}}

\def\rbvp{(\ref{eq:rbvp})}

\begin{document}

\title{Conformal invariance in random cluster models. \\
I. Holomorphic fermions in the Ising model.}
\author{Stanislav Smirnov}
%\address{}
\email{Stanislav.Smirnov@math.unige.ch}
\date{\today}
%\subjclass{}
%\begin{abstract}
%\end{abstract}

\maketitle

\pagestyle{myheadings}
\markboth{Stanislav Smirnov}{Conformal invariance in random cluster models. I.}

\section{Introduction}\label{sec:intro}

It is widely believed that many planar lattice models at the critical temperature
are conformally invariant in the scaling limit.
In particular, 
the Ising model is often cited as a classical example
of conformal invariance which is used in deriving many of its properties.

To the best of our knowledge no mathematical proof of this assertion
has ever been given.
Moreover,  
most of the physics arguments 
concern rectangular domains only (like a plane or a strip) 
or do not take boundary conditions into account.
Thus they give (often unrigorous) justification 
only of the {\it M\"obius invariance} of the scaling limit,
arguably a much weaker property
than full {\it conformal invariance}.
Of course, success of conformal field theory
methods in describing the Ising model provides
some evidence for the conformal invariance,
but it does not offer an explanation or a proof of the latter.

It seems that ours is the first paper, where 
actual {\it conformal invariance} of some observables
for the Ising model at criticality
(in domains with appropriate boundary conditions)
is established.
Our methods are different from those employed before,
and allow us to obtain sharper versions of some of the known results.
Moreover they allow the construction of
conformally invariant observables in domains
with complicated boundary conditions
and on Riemann surfaces.
Ultimately we will construct conformally invariant
scaling limits of interfaces (random cluster boundaries)
and identify them with Schramm's \slee~ curves
and related loop ensembles.
These extensions will be discussed in the sequels \cite{smirnov-fk2,smirnov-fk3}.
Though one can argue whether the scaling limits of interfaces in the Ising model
are of physical relevance,
their identification opens possibility for computation of
correlation functions and other objects of interest in physics.

We consider the Fortuin-Kasteleyn random cluster representation
of the Ising model on the square lattice $\mesh\ZZ^2$
at the critical temperature.
This representation, briefly reviewed below,
studies \emph{random clusters}, which are clusters of the critical percolation 
performed on the Ising \emph{spin clusters} at the critical temperature.
The spin correlations can be easily related to connectivity properties in the new model.
Every configuration can be described by a collection of interfaces
(between random clusters and dual random clusters)
which are disjoint loops that
fill all the edges of the medial lattice.

As a conformally invariant observable we
construct a ``discrete holomorphic fermion''. 
In a simply connected domain $\Om$ with two boundary points $a$ and $b$
we introduce Dobrushin boundary conditions, which
enforce the existence (besides many loops) of an interface running from $a$ to $b$,
see Figure~\ref{fig:loops}.
We show that the expectation that this interface passes through a point $z$
taken with fermionic weight
(i.e. a passage in the same direction but with a $2\pi$ twist has a relative
weight $-1$, whereas a passage in the opposite direction with a
counterclockwise $\pi$ twist has a relative weight $-i$ -- see Figure~\ref{fig:weight})
is a discrete holomorphic function of $z$.
Moreover, as the step of the lattice goes to zero,
this expectation, when appropriately normalized,
converge to a conformally covariant scaling limit, namely $\sqrt{\Phi'}$, where $\Phi$ is the conformal
map of $\Om$ to a horizontal strip.

The approach is set up for a random cluster model
with a general value of the parameter $q\in[0,4]$,
and a parafermion observable of spin $\spin=1-\frac2\pi\arccos({\sqrt{q}}/2)$,
conjecturally converging to $\br{\Phi'}^\spin$ in the general case.
The Ising case corresponds to $q=2$ and $\spin=1/2$.
Besides a priori estimates (which are well-known in the Ising case),
we make essential use of the Ising-specific properties in two places:
to establish discrete analyticity of an observable,
and to show that being a solution of the discrete Riemann boundary
value problem, it converges to its continuum counterpart.
For the latter we see possibilities for a proof in the general case,
albeit more difficult.
So it seems that the only essential obstacle to
proving conformal invariance of {\it all} random cluster models
lies in establishing discrete analyticity of the observable concerned.
For the Ising case this is done by proving discrete analogues of Cauchy-Riemann relations,
where partial results can be obtained for all random cluster models.

The two sequels  \cite{smirnov-fk2,smirnov-fk3}
are concerned with the construction (on the basis of one observable)
of conformally invariant scaling limits of one interface 
and full collection of interfaces respectively.
They are, more or less, applicable to all random cluster models for which
conclusions of this first part, in particular Theorem~\ref{thm:fermi},
can be established.
In the Ising case the law of one interface converges
to that of the Schramm-L\"owner Evolution with $\kappa=16/3$.

These results were announced and the proofs were sketched  in \cite{smirnov-icm},
where one can find some of the ideas leading to our approach.
Another notable case when this approach (or rather a parallel one)
works is the usual spin representation of the Ising model
at critical temperature, leading to a similar observable
(related to a conformal map to a halfplane), 
and to Schramm-L\"owner Evolution with $\kappa=3$.
On a rectangular lattice, exactly the same notion of discrete analyticity
arises. This will be discussed in a separate paper.

Similar observables were constructed before by Kenyon \cite{Kenyon-conformal} 
and by the author \cite{\smirnovcras,smirnov-perc}.
The work of Kenyon concerned dimers on the square lattice
(domino tilings) and so
by the Temperley bijection gave conformally invariant 
observables for the Uniform Spanning Tree 
(corresponding to the random cluster model with $q=0$) and Loop Erased Random Walk. 
Kenyon's considerations are close in spirit to ours,
in fact repeating his constructions for the Fisher lattice one is
led to similar observables. 
Since \slee~ was not available at the moment, the identification of interfaces
had to wait till the work \cite{\lswust} of Lawler, Schramm and Werner.
Nevertheless Kenyon was able to rigorously determinate several exponents and dimensions, 
and some of his results go beyond the reach of \slee~ machinery.
We constructed \cite{\smirnovcras,smirnov-perc} conformally
invariant observables for the critical site percolation on the triangular lattice
which also bear some similarity to ones in the current paper.
Unfortunately that proof is restricted to the triangular lattice,
so the question of conformal invariance 
remains open for the percolation on the square lattice
(which corresponds to the random cluster model with $q=1$).

The paper is organized as follows.
In  
Section~\ref{sec:results} we state our theorem.
We start the proof by introducing a new notion of discrete analyticity
in Section~\ref{sec:discrete}, and then show that it is satisfied
by an Ising model observable, 
which we construct in Section~\ref{sec:fermi}.
Finally in  
Section~\ref{sec:limit} we show that the discrete
observable has a conformally covariant scaling limit.
In the proof we use an a priori estimate for the Ising model,
which follows from (a weak form) of known magnetization estimates;
this is discussed in Appendix~\ref{sec:apriori}.
Some of the more technical results about discrete harmonic functions
are reviewed in Appendix~\ref{sec:dharm}.

\subsection*{Acknowledgments}
The author gratefully acknowledges support of the Swiss National Science Foundation.
Much of the work was completed while the author was a Royal Swedish Academy of Sciences Research Fellow 
supported by a grant from the Knut and Alice Wallenberg Foundation.

Existence of a discrete holomorphic function in the Ising spin model
which has potential to imply convergence of interfaces to \sle3
was first noticed by Rick Kenyon and the author in 2002 based on 
the dimer techniques applied to the Fisher lattice.
However at the moment the Riemann Boundary Value Problem 
(similar to the one arising in this paper)
seemed beyond reach.

When this manuscript was written, the author learned 
that John Cardy and Valentina Riva were preparing a preprint
discussing the physics consequences of the fact
that (the classical version) of discrete analyticity 
holds for the function (\ref{eq:fdef}) restricted to edges.
However, their work \cite{\cardyriva}
does not address the boundary conditions
and the convergence to the scaling limit.

The author is indebted to Lennart Carleson for introducing him to this area, as well
as for encouragement and advice.
The author is grateful to Michael Aizenman,
Ilia Binder, John Cardy, Lincoln Chayes,
Geoffrey Grimmett, David Kazhdan, Nikolai Makarov and Andrei Okounkov
for fruitful discussions of this manuscript.

\begin{figure}
{\def\bs{\circle*{3}}\def\ws{\circle{3}}
\def\ne{\qbezier(-0.,3)(0.5,0.5)(3,0)}
\def\nw{\qbezier(-0.,3)(-0.5,0.5)(-3,0)}
\def\se{\qbezier(-0.,-3)(0.5,-0.5)(3,0)}
\def\sw{\qbezier(-0.,-3)(-0.5,-0.5)(-3,0)}
\def\bne{\qbezier(-0.3,3)(0.2,0.5)(3,-0)}
\def\bnw{\qbezier(-0.3,3)(-0.8,0.5)(-3,0)}
\def\bse{\qbezier(-0.3,-3)(0.2,-0.5)(3,0)}
\def\bsw{\qbezier(-0.3,-3)(-0.8,-0.5)(-3,0)}
\def\uu{{\thicklines\line(1,1){10}}}
\def\dd{{\thicklines\line(1,-1){10}}}
\def\uw{\line(1,1){10}}
\def\dw{\line(1,-1){10}}
\unitlength=0.8mm
\begin{picture}(115,85)(-10,-10)
\put(-5,15)\uw\put(-5,35)\uw\put(-5,55)\uw\put(05,45)\uw\put(15,-5)\uw
\put(25,05)\uw\put(25,25)\uw\put(35,-5)\uw\put(35,15)\uw\put(35,35)\uw\put(55,-5)\uw
\put(55,15)\uw\put(55,35)\uw\put(65,25)\uw\put(75,-5)\uw\put(75,15)\uw\put(85,05)\uw
\put(-5,15)\dw\put(-5,35)\dw\put(-5,55)\dw\put(05,05)\dw\put(05,45)\dw\put(05,65)\dw
\put(15,35)\dw\put(25,05)\dw\put(25,25)\dw\put(35,15)\dw\put(45,05)\dw\put(45,25)\dw
\put(55,15)\dw\put(55,55)\dw\put(65,25)\dw\put(65,05)\dw\put(75,55)\dw\put(75,35)\dw
\put(15,45)\uu\put(05,15)\uu\put(15,05)\uu\put(25,35)\uu\put(25,55)\uu\put(35,45)\uu
\put(45,35)\uu\put(45,55)\uu\put(55,25)\uu\put(65,35)\uu\put(65,55)\uu\put(75,05)\uu
\put(85,15)\uu\put(85,35)\uu\put(35,25)\uu
\put(05,15)\dd\put(05,35)\dd\put(15,25)\dd\put(15,45)\dd\put(25,55)\dd\put(35,65)\dd
\put(45,15)\dd\put(45,35)\dd\put(45,55)\dd\put(55,65)\dd\put(65,15)\dd
\put(65,55)\dd\put(75,45)\dd\put(75,65)\dd\put(85,55)\dd\put(85,35)\dd\put(15,65)\dd
{\linethickness{1.5pt}
\put(10,65){\line(0,-1){2}}\put(10,47){\line(0,-1){4}}\put(20,37){\line(0,-1){4}}\put(20,57){\line(0,-1){4}}
\put(30,37){\line(0,-1){4}}\put(30,27){\line(0,-1){4}}\put(40,47){\line(0,-1){4}}\put(40,37){\line(0,-1){4}}
\put(40,27){\line(0,-1){4}}\put(50,47){\line(0,-1){4}}\put(50,27){\line(0,-1){4}}\put(60,27){\line(0,-1){4}}
\put(60,17){\line(0,-1){4}}\put(70,37){\line(0,-1){4}}\put(70,17){\line(0,-1){4}}\put(70,7){\line(0,-1){4}}
\put(80,37){\line(0,-1){4}}\put(80,17){\line(0,-1){4}}\put(80,7){\line(0,-1){4}}\put(90,27){\line(0,-1){4}}
\put(90,17){\line(0,-1){4}}\put(73,0){\line(1,0){4}}\put(63,10){\line(1,0){4}}\put(83,10){\line(1,0){4}}
\put(33,20){\line(1,0){4}}\put(53,20){\line(1,0){4}}\put(63,20){\line(1,0){4}}\put(83,20){\line(1,0){4}}
\put(93,20){\line(1,0){2}}\put(23,30){\line(1,0){4}}\put(33,30){\line(1,0){4}}\put(63,30){\line(1,0){4}}
\put(83,30){\line(1,0){4}}\put(13,40){\line(1,0){4}}\put(33,40){\line(1,0){4}}\put(43,40){\line(1,0){4}}
\put(73,40){\line(1,0){4}}\put(13,50){\line(1,0){4}}\put(43,50){\line(1,0){4}}\put(13,60){\line(1,0){4}}
\put(43,30){\line(1,0){4}}\put(73,10){\line(1,0){4}}
\put(10,60)\bne\put(10,40)\bne\put(20,30)\bne\put(30,20)\bne\put(50,20)\bne\put(60,10)\bne
\put(70,00)\bne\put(70,10)\bne\put(80,30)\bne
\put(20,50)\bnw\put(30,30)\bnw\put(40,20)\bnw\put(40,40)\bnw\put(50,40)\bnw\put(60,20)\bnw
\put(70,30)\bnw\put(80,00)\bnw\put(90,10)\bnw\put(90,20)\bnw\put(80,10)\bnw\put(40,30)\bnw
\put(20,60)\bsw\put(20,40)\bsw\put(50,30)\bsw\put(50,50)\bsw\put(70,10)\bsw\put(70,20)\bsw
\put(80,40)\bsw\put(90,30)\bsw
\put(10,50)\bse\put(30,30)\bse\put(30,40)\bse\put(40,30)\bse\put(40,40)\bse\put(40,50)\bse
\put(60,20)\bse\put(60,30)\bse\put(70,40)\bse\put(80,10)\bse\put(80,20)\bse\put(90,20)\bse
}
\put(10,00)\ne\put(30,00)\ne\put(50,00)\ne\put(00,10)\ne\put(40,10)\ne\put(50,10)\ne
\put(20,20)\ne\put(70,20)\ne\put(00,30)\ne\put(10,30)\ne\put(50,30)\ne\put(20,40)\ne
\put(80,40)\ne\put(00,50)\ne\put(30,50)\ne\put(50,50)\ne\put(60,50)\ne\put(70,50)\ne
\put(80,50)\ne\put(10,10)\ne
\put(20,00)\nw\put(40,00)\nw\put(60,00)\nw\put(20,10)\nw\put(30,10)\nw\put(10,20)\nw
\put(80,20)\nw\put(60,30)\nw\put(30,40)\nw\put(70,40)\nw\put(90,40)\nw\put(10,50)\nw
\put(40,50)\nw\put(60,40)\nw
\put(10,10)\sw\put(40,10)\sw\put(50,10)\sw\put(60,10)\sw\put(20,20)\sw\put(30,20)\sw
\put(50,20)\sw\put(10,30)\sw\put(20,30)\sw\put(80,30)\sw\put(10,40)\sw\put(60,50)\sw
\put(70,50)\sw\put(80,50)\sw\put(90,50)\sw\put(10,60)\sw\put(40,60)\sw\put(60,60)\sw
\put(80,60)\sw\put(30,50)\sw
\put(20,10)\se\put(30,10)\se\put(00,20)\se\put(10,20)\se\put(40,20)\se\put(70,30)\se
\put(00,40)\se\put(50,40)\se\put(20,50)\se\put(00,60)\se\put(30,60)\se\put(60,40)\se
\put(50,60)\se\put(70,60)\se
\put(00,17){\line(0,-1){4}}\put(00,37){\line(0,-1){4}}\put(00,57){\line(0,-1){4}}\put(10,07){\line(0,-1){4}}
\put(10,17){\line(0,-1){4}}\put(10,27){\line(0,-1){4}}\put(10,37){\line(0,-1){4}}\put(10,47){\line(0,-1){4}}
\put(20,07){\line(0,-1){4}}\put(20,17){\line(0,-1){4}}\put(20,27){\line(0,-1){4}}\put(20,47){\line(0,-1){4}}
\put(30,07){\line(0,-1){4}}\put(30,17){\line(0,-1){4}}\put(30,47){\line(0,-1){4}}\put(30,57){\line(0,-1){4}}
\put(40,07){\line(0,-1){4}}\put(40,17){\line(0,-1){4}}\put(40,57){\line(0,-1){4}}\put(50,07){\line(0,-1){4}}
\put(50,17){\line(0,-1){4}}\put(50,37){\line(0,-1){4}}\put(50,57){\line(0,-1){4}}\put(60,07){\line(0,-1){4}}
\put(60,37){\line(0,-1){4}}\put(60,57){\line(0,-1){4}}\put(60,47){\line(0,-1){4}}\put(70,27){\line(0,-1){4}}
\put(70,47){\line(0,-1){4}}\put(70,57){\line(0,-1){4}}\put(80,27){\line(0,-1){4}}\put(80,47){\line(0,-1){4}}
\put(80,57){\line(0,-1){4}}\put(90,47){\line(0,-1){4}}\put(10,57){\line(0,-1){4}}
\put(13,00){\line(1,0){4}}\put(33,00){\line(1,0){4}}\put(53,00){\line(1,0){4}}\put(03,10){\line(1,0){4}}
\put(13,10){\line(1,0){4}}\put(23,10){\line(1,0){4}}\put(33,10){\line(1,0){4}}\put(43,10){\line(1,0){4}}
\put(53,10){\line(1,0){4}}\put(03,20){\line(1,0){4}}\put(13,20){\line(1,0){4}}\put(23,20){\line(1,0){4}}
\put(43,20){\line(1,0){4}}\put(03,30){\line(1,0){4}}\put(13,30){\line(1,0){4}}\put(53,30){\line(1,0){4}}
\put(53,40){\line(1,0){4}}\put(63,40){\line(1,0){4}}\put(83,40){\line(1,0){4}}\put(03,50){\line(1,0){4}}
\put(03,40){\line(1,0){4}}\put(23,40){\line(1,0){4}}\put(73,30){\line(1,0){4}}\put(73,20){\line(1,0){4}}
\put(23,50){\line(1,0){4}}\put(33,50){\line(1,0){4}}\put(53,50){\line(1,0){4}}\put(63,50){\line(1,0){4}}
\put(73,50){\line(1,0){4}}\put(83,50){\line(1,0){4}}\put(03,60){\line(1,0){4}}\put(33,60){\line(1,0){4}}
\put(53,60){\line(1,0){4}}\put(73,60){\line(1,0){4}}
\put(09,68){$a$}\put(98,18){$b$}
\put(10,65){\circle*{2}}
\put(95,20){\circle*{2}}
\put(15,05)\bs\put(35,05)\bs\put(55,05)\bs\put(75,05)\bs
\put(15,25)\bs\put(35,25)\bs\put(55,25)\bs
\put(75,25)\bs\put(95,25)\bs
\put(15,45)\bs\put(35,45)\bs\put(55,45)\bs
\put(75,45)\bs\put(95,45)\bs
\put(15,65)\bs\put(35,65)\bs\put(55,65)\bs\put(75,65)\bs
\put(05,15)\bs\put(25,15)\bs\put(45,15)\bs\put(65,15)\bs\put(85,15)\bs
\put(05,35)\bs\put(25,35)\bs\put(45,35)\bs\put(65,35)\bs\put(85,35)\bs
\put(05,55)\bs\put(25,55)\bs\put(45,55)\bs\put(65,55)\bs\put(85,55)\bs
\put(05,05)\ws\put(25,05)\ws\put(45,05)\ws\put(65,05)\ws\put(85,05)\ws
\put(05,25)\ws\put(25,25)\ws\put(45,25)\ws\put(65,25)\ws\put(85,25)\ws
\put(05,45)\ws\put(25,45)\ws\put(45,45)\ws\put(65,45)\ws\put(85,45)\ws
\put(05,65)\ws
\put(-5,15)\ws\put(15,15)\ws\put(35,15)\ws\put(55,15)\ws\put(75,15)\ws
\put(-5,35)\ws\put(15,35)\ws\put(35,35)\ws\put(55,35)\ws\put(75,35)\ws
\put(-5,55)\ws\put(15,55)\ws\put(35,55)\ws\put(55,55)\ws\put(75,55)\ws
\put(95,15)\ws
\put(15,-5)\ws\put(35,-5)\ws\put(55,-5)\ws\put(75,-5)\ws
\end{picture}}
\caption{Loop representation
of the random cluster Ising model.
Weight of the configuration is proportional
to $(\sqrt{q})^{\mathrm{\#~loops}}$, with $q=2$.
The sites of the original Ising lattice are colored in black,
while the sites of the dual lattice are colored in white.
Loops separate clusters from dual clusters,
which are also pictured, the former in bold.
Under Dobrushin boundary conditions
besides a number of loops there is an interface running from $a$ to $b$, which is drawn in bold.
}
\label{fig:loops}
\end{figure}
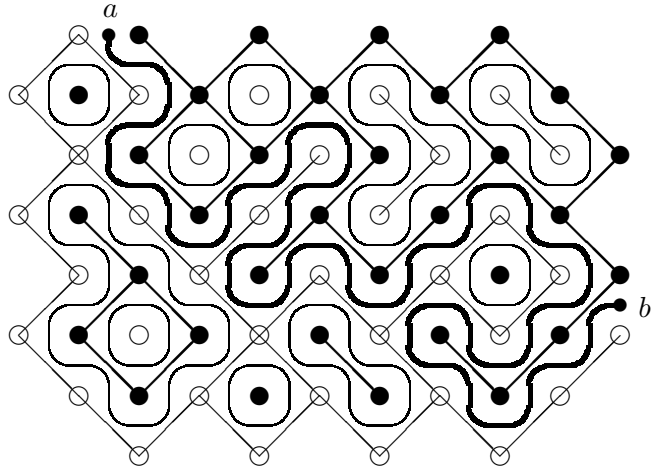

\section{Statement of results}\label{sec:results}

We work with the Fortuin-Kasteleyn random cluster model
with a particular emphasis on the critical Ising case,
corresponding to parameter values $q=2$ and $p=\sqrt{2}/(\sqrt{2}+1)$.
For a general introduction to the Ising and random cluster models
consult the books \cite{\baxterbook,\grimmettbookrc,\mccoywubook}.

The random cluster measure on a graph (a simply connected domain $\Om$ on the square lattice in our case)
is a probability measure on edge configurations (when each edge is declared either open or closed),
such that the probability of a configuration is proportional to
$$p^{\mathrm{\#~open~edges}}~(1-p)^{\mathrm{\#~closed~edges}}~q^{\mathrm{\#~clusters}}~,$$
where clusters are maximal subgraphs connected by open edges.
The two parameters are edge-weight $p\in[0,1]$ and cluster-weight $q\in[0,\infty)$,
with $q\in[0,4]$ being of interest to us.
For a square lattice (or in general any planar graph) to every configuration
one can prescribe a random cluster configuration on the dual graph, such that
every open edge is intersected by a dual closed edge and vice versa.
See Figure~\ref{fig:loops} for two dual configurations
with open edges pictured.
It turns out that the probability of a dual configuration becomes proportional
to
$$p_*^{\mathrm{\#~dual~open~edges}}~(1-p_*)^{\mathrm{\#~dual~closed~edges}}~q^{\mathrm{\#~dual~clusters}}~,$$
with the dual to $p$ value $p_*=p_*(p)$ satisfying $p_*/(1-p_*)=q(1-p)/p$.
For $p=p_{sd}:=\sqrt{q}/(\sqrt{q}+1)$ the dual value coincides with the original one:
one gets $p_{sd}=(p_{sd})_*$ and so the model is self-dual.
It is conjectured that this is also the critical value of $p$,
which was only proved for $q=1$ (percolation), $q=2$ (Ising) 
and $q>25.72$.
For these and other properties of the random cluster models
consult Grimmett's monograph \cite{\grimmettbookrc}.

We will work with the loop representation,
which is perhaps the easiest to visualize.
The cluster configurations can be represented as (Hamiltonian)
loop configurations on the medial lattice
(a square lattice which has as vertices edge centers of the original lattice),
with loops representing
interfaces between
cluster and dual clusters -- see Figure~\ref{fig:loops}.
It is well-known that for $p=p_{sd}$
the probability of a configuration
is proportional to 
$$\br{\sqrt{q}}^{\mathrm{\#~loops}}~,$$
with $q=2$ in the Ising case.

We introduce Dobrushin boundary conditions:
wired on the counterclockwise arc $ba$
(meaning that all edges along the arc are open)
and dual-wired on the counterclockwise arc $ab$
(meaning that all dual edges along the arc are open,
or equivalently all primal edges orthogonal to the arc are closed)
- see Figure~\ref{fig:loops}.
For the loop representation this reduces to introducing
two vertices with odd number of edges:
a source $a$ and a sink $b$.
Then besides a number of closed loop interfaces
there is a unique interface running from $a$ to $b$, 
which separates the cluster containing the arc $ab$ from the dual cluster containing the arc $ba$.
See Figures~\ref{fig:loops},~\ref{fig:perestroika1} 
for typical configurations.

Note that Dobrushin boundary conditions are usually formulated for
the spin Ising model and amount to setting plus and minus spin boundary conditions
on two arcs correspondingly, thus creating an interfaces between two spin clusters.
Since we need an interface between two random clusters, we formulated similar conditions in the random cluster setting.
Our version of Dobrushin boundary conditions is equivalent to setting plus boundary conditions
on one arc and free on the other in the spin setting.

The model makes sense for $q=0$ as well, and
is equivalent to the uniform spanning tree model.
Indeed, setting $q=0$  prohibits loops, so we consider configurations containing
only an interface from $a$ to $b$, which are weighted uniformly
since the number of open edges is always the same.
Those configurations are easily seen to be equivalent
to spanning trees on the original spin lattice, rooted on the arc $ab$. We are mostly interested in the Ising case of $q=2$ and $p=p_{sd}=\sqrt{2}/(\sqrt2+1)$.
Though it is known that this value of $p$ is critical, we won't use it in the proof.
For other values of $q$ our proof works to large extent,
also for the self dual  value  $p=p_{sd}$,
and in principle one can try to use this in establishing its criticality.

Define spin by $\spin:=1-\frac2\pi\arccos(\sqrt{q}/2)$.
Note that for the Ising case $\spin=\frac12$.
Let $F(z)$ be the expectation that the interface $\ga$ passes through a point $z$
taken with a complex weight:
$$
F(z)~:=~
\Exp{\chi_{z\in\ga(\omega)}\,\cdot
\,\exp\br{- i\spin\,\wi(\ga,b\to z)}}~.
$$
Here $\wi$ denotes the winding or (the total turn) of $\ga$ from $b$ to $z$,
measured in radians.
For the Ising case an additional $2\pi$ turn of the curve before reaching $z$ changes
the weight by a factor of $-1$, see Figure~\ref{fig:weight}.
The formula above gives $F$ at the edge centers (of the medial lattice, where the loop representation is defined), 
and we extend
it to all of $\Om$ in a standard piecewise constant way.
Exact definition can be found below.

We start by proving in Section~\ref{sec:fermi} the following
\begin{prop}\label{prop:prehol}
For the Ising model in a given lattice domain the function $F(z)$ is discrete holomorphic
and satisfies a discrete analogue of the Riemann Boundary Value Problem (with $\spin=\frac12$)
\begin{equation}\label{eq:rbvp}
\IM \br{F(z)\,{{\mathrm{tangent}}(z)}^\spin}~=~0~.
\end{equation}
\end{prop}
The continuum problem is solved by $f=\br{\Phi'}^\sigma$,
where $\Phi$ is the conformal map of $\Om$ to a horizontal strip,
with $a$ and $b$ mapped to the ends.
In our case normalization will produce a trip of width $2$.
After some technicalities in Section~\ref{sec:limit}
we show that $F$ converges to its continuum counterpart:
\begin{thm}\label{thm:fermi}
Suppose that as the lattice mesh $\mesh_j$ goes to zero,
the discrete domains $\Om_j$  on the lattices $\mesh_j\ZZ^2$
(with points $a_j$, $b_j$ on the boundary)
converge to the domain 
$\Om$ (with points $a$, $b$ on the boundary)
in the Carath\'eodory sense.
Then for the Ising model 
the corresponding functions $F_j=F(z,\Om_j,a_j,b_j,\mesh_j\ZZ^2)$ converge
uniformly away from the boundary:
\begin{equation}
{\mesh_j}^{-\spin}\,{F_j}~\unif~ f = \br{\Phi'}^\spin~,\label{eq:convf}
\end{equation}
where $\spin=1/2$.
\end{thm}

\begin{rem}\label{rem:cara}
Carath\'eodory convergence is defined as
convergence of normalized Riemann uniformization maps
on compact subsets.

Namely we fix a point $w\in\Omega$, and let $\phi$ (or $\phi_j$) be conformal maps
from the unit disk $\DD$ to $\Omega$ (or $\Omega_j$) such that
points $0,1,\zeta$ (or $0,1,\zeta_j$)
are mapped to $w,a,b$ (or $w,a_j,b_j$) or corresponding prime ends.
We say that $\Omega_j$ converge to $\Omega$ if
$\phi_j$ converge to $\phi$ inside $\DD$ and $\zeta_j$ tends to $\zeta$.

It is easy to see that Hausdorff convergence of the boundaries implies Carath\'eodory convergence and
that solution to the Riemann boundary value problem (\ref{eq:rbvp}) for $z$ inside $\Om$, 
being defined in the terms of Riemann maps,
is uniformly continuous as a function of $\Om$ in the topology of Carath\'eodory convergence.
\end{rem}

A variation of our proof seems to work for $q=0$ as well, and
most of it can be worked out for other values of $q$, 
though sometimes in a different way.
Hopefully the missing part of the discrete analyticity statement
will be worked out someday:
\begin{conj}\label{conj}
Proposition~\ref{prop:prehol}
(with an appropriate, possibly approximate, discrete analyticity)
and Theorem~\ref{thm:fermi}
hold for all values of $q\in[0,4]$.
\end{conj}
%This theorem is used in the sequel \cite{smirnov-fkisingii} to identify
%scaling limits of interfaces with \sle{16/3} curves.

\section{Discrete analyticity revisited}\label{sec:discrete}
 We will identify lines through the origin with unit vectors
(complex numbers) belonging to them.
For a line $\ell$ or equivalently a  vector $\alpha\in\ell$
we denote by $\Proj{F,\ell}=\Proj{F,\alpha}$
the orthogonal projection of a complex number $F$ on the line $\ell$.
Note that for a unit vector $\alpha$
\begin{equation}\label{eq:proj}
\Proj{F,\alpha}~=~\alpha\,\RE\br{\bar\alpha F}~=~\br{F+\alpha^2\bar F}/2~.
\end{equation}
Consider the square lattice $\mesh\ZZ^2$ (possibly rotated).
By a lattice domain $\Om$ we mean some collection of vertices joint by edges
such that all vertices have even number of edges.
In our application we will also allow half-edges
(usually two, their middle ends will become  
 the source $a$ and the sink $b$).

By distance between two points (when speaking of moduli of continuity of functions, etc.) inside $\Om$ we will mean the distance in the inner metric.

If for some vertex all four edges are present, we call it
an {\em interior vertex},
while vertices with two edges we call {\em boundary vertices}.
If for a square all four vertices are interior,
we call it an {\em interior square}.
To the lattice domain $\Om$
we associate a planar domain $\tilde\Om$ 
which is the union of all interior squares.
We will assume that those domains are connected and simply connected.

Color the lattice squares in a chessboard fashion.
We orient every edge $e$, turning
it into a unit vector (or a complex number) $\vec e$ with the orientation chosen so that 
the white square is on the left and the black one on the right.
Then to the edge $e$ prescribe a line $\ell(e)$ in the complex plane
which passes through the origin and the square root of the complex conjugate
of the vector $\vec e$, considered as a complex number
(note that the choice of the square root is not important).

Without loss of generality we can assume that the lattice
edges are parallel to the coordinate axis 
(otherwise all lines are rotated by a fixed angle). 
Then horizontal edges correspond to the lines 
with argument (defined up to $\pi$) $0$ or $\pi/2$
(in the chessboard order),
whereas vertical edges correspond to $\pi/4$ or $3\pi/4$.
See Figure~\ref{fig:int}.

A given vertex $v$ has 4 neighboring edges.
If we go around $v$ counterclockwise, the line corresponding to the neighbor 
with each step is rotated counterclockwise by $\pi/4$
(so the full turn corresponds to a rotation by $\pi$,
which preserves the line but reverses directions).

\begin{defi}
We say that a function $F$ defined on vertices is {\em preholomorphic} or
{\em discrete analytic} in a domain
$\Om$ if for every edge $e\in\Om$ orthogonal projections of the values of $F$
at its endpoints on the line $\ell(e)$ coincide.
We will denote this common projection by $F(e)$.
\end{defi}

\begin{rem}
In the complex plane {\em holomorphic} (i.e. having a complex derivative)
and {\em analytic} (i.e. admitting a power series expansion) functions are the same,
so the terms are often interchanged.
Though the term {\em discrete analytic} is in wide use,
in discrete setting there are no power expansions, so
it would be more appropriate to speak of
{\em discrete holomorphic} (or {\em discrete regular}) functions.
We prefer the term {\em preholomorphic}, which was common at one point, but seems to 
have gone out of use.
\end{rem}

\begin{rem}\label{rem:usualhol}
The more commonly used discrete analyticity condition asks for
the discrete version
of the Cauchy-Riemann equations 
$$\partial_{i\alpha}F=i\partial_{\alpha}F$$
to be satisfied.
Namely, for a lattice square with 
corner vertices $NW$, $NE$, $SE$, $SW$ 
(starting from the upper left and going clockwise)
one has 
\begin{equation}
F(NW)-F(SE)~=~i\,\br{F(NE)-F(SW)}~.
\end{equation}
It is easy to check
that our property implies the more common one on vertices, but does not follow from it.
Moreover, our property is equivalent to the more common property for the 
function restricted to horizontal edges.
\end{rem}

%\begin{lem}\label{lem:usualhol}
%Suppose that some lattice square has 
%corners vertices $NW$, $NE$, $SE$, $SW$ 
%starting from the upper left and going clockwise.
%Then $F$ satisfies the discrete version of the Cauchy-Riemann equations:
%\begin{equation}
%F(NW)-F(SE)~=~i\,\br{F(NE)-F(SW)}~.
%\end{equation}
%\end{lem}

%\begin{proof}
%\end{proof}

%\begin{rem}The discrete analytic function for the Ising model will be
%the expectation (with complex weight) that interface
%passes through a vertex.\end{rem}

\begin{defi}\label{def:rbvp}
We will say that a preholomorphic function $F$ solves
the Riemann Boundary Value Problem $\rbvp$
in the domain $\Om$,
if for every boundary vertex $v$ with two edges,
projections of $F(v)$ on the lines corresponding
to these edges have the same modulus.
\end{defi}

\begin{rem}
Indeed, since projections of $F(v)$ on lines corresponding to two edges
coincide, $F(v)$ belongs to the bisector of those lines.
Equivalently the value of $F$ at every boundary vertex $v$
is parallel to the reciprocal of the square root of the tangent vector $\tau(v)$
(or rather a discrete approximation -- the vector orthogonal
to the bisector of the angle between two edges from $v$).
This is a discrete analogue of the 
Riemann Boundary Value Problem (\ref{eq:rbvp})
and our main goal will be to show that solutions to the discrete
problem in the limit solve the continuum one.
\end{rem}

We will solve the problem \rbvp~ by ``integrating'' the square of $F$,
which is not so easy since $F^2$ is no longer preholomorphic.

\begin{lem}\label{lem:phi}
Let $F$ be preholomorphic in a domain $\Om$.
Then up to a constant there is a unique function $\HH=\HH_F$
defined on the lattice squares in $\Om$ or adjacent to $\Om$
and such that for any two adjacent squares,
say black $B$ and white $W$ separated by the edge $e$, one has
\begin{equation}\label{eq:phi}
\HH(B)-\HH(W)~=~|F(e)|^2~.
\end{equation}
\end{lem}
In applications we will chose the constant 
so that $\HH$ is zero on a square immediately below $b$.
Note that $\HH$ is defined on the dual lattice to $\mesh \ZZ^2$,
which is the usual case for discrete derivatives or primitives.

\begin{rem}
Values of the argument of $F$ on edges are such
that the function $2 \mesh \HH$ is a discrete analogue
of the indefinite integral $\IM \int F^2 dz$,
which we will establish
in Appendix~\ref{sec:subharm}, equation (\ref{eq:hintf}):
if $u$ and $v$ are centers of two squares with a common corner $z$,
then
$$2\,\mesh\,(H(v)-H(u))=\IM \br{F(z)^2\,(v-u)}.$$
\end{rem}

\begin{proof}
It is sufficient to check that when one goes around an interior vertex,
increments of $\HH$ add up to zero.
Suppose that the edge neighbors of the vertex $v$ 
are $E$, $S$, $W$, $N$ in clockwise order starting from the right.
Then the sum of increments when we go around $v$ is
\begin{equation}
\pm|F(E)|^2\mp|F(N)|^2\pm|F(W)|^2\mp|F(S)|^2~,
\label{eq:incr}\end{equation}
with signs depending on the choice of chessboard coloring.
By construction $\ell(E)\perp\ell(W)$, $\ell(N)\perp\ell(S)$.
Since $F(E)$, $F(W)$, $F(N)$, $F(S)$ are orthogonal
projections of $F(z)$ on the corresponding lines,
by Pythagoras theorem
$$|F(v)|^2=|F(E)|^2+|F(W)|^2=|F(N)|^2+|F(S)|^2~.$$
Thus the sum of the increments is equal to
$$\pm(|F(E)|^2+|F(W)|^2-|F(N)|^2-|F(S)|^2)=\pm(|F(v)|^2-|F(v)|^2)=0~,$$
and indeed vanishes.
\end{proof}
 
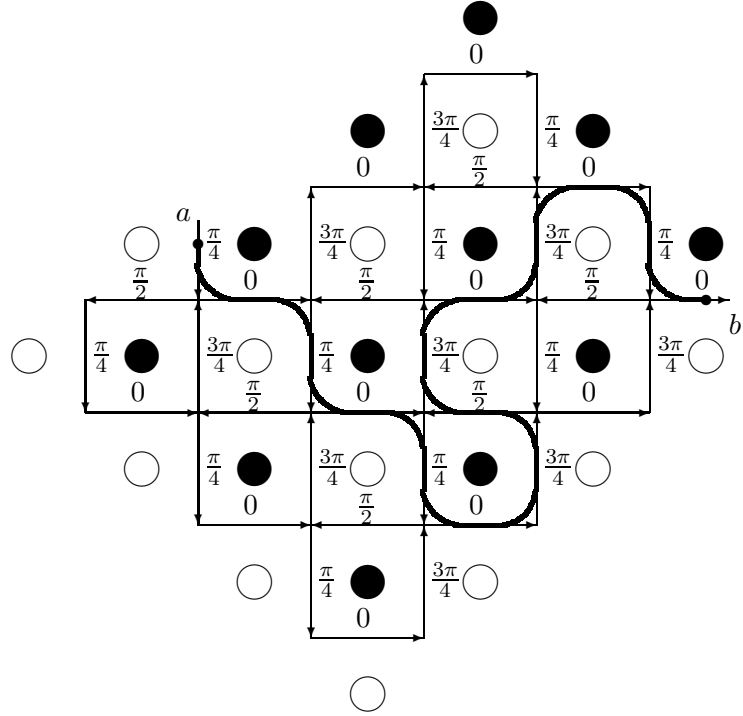
\begin{figure}
{\def\bs{\circle*{3}}\def\ws{\circle{3}}
\def\pa{$0$}\def\pb{$\frac\pi2$}\def\pc{$\frac{3\pi}4$}\def\pd{$\frac\pi4$}
\def\ne{\qbezier(-0.15,3)(0.5,0.5)(3,0)}
\def\nw{\qbezier(-0.15,3)(-0.5,0.5)(-3,0)}
\def\se{\qbezier(-0.15,-3)(0.5,-0.5)(3,0)}
\def\sw{\qbezier(-0.15,-3)(-0.5,-0.5)(-3,0)}
\unitlength=1.5mm
\begin{picture}(70,70)(-10,-10)
\put(20,0){\vector(1,0){10}}
\put(10,10){\vector(1,0){10}}\put(30,10){\vector(-1,0){10}}\put(30,10){\vector(1,0){10}}
\put(0,20){\vector(1,0){10}}\put(20,20){\vector(-1,0){10}}\put(20,20){\vector(1,0){10}}
 \put(40,20){\vector(-1,0){10}}\put(40,20){\vector(1,0){10}}
\put(10,30){\vector(-1,0){10}}\put(10,30){\vector(1,0){10}}\put(30,30){\vector(-1,0){10}}
 \put(30,30){\vector(1,0){10}}\put(50,30){\vector(-1,0){10}}\put(50,30){\vector(1,0){7}}
\put(20,40){\vector(1,0){10}}\put(40,40){\vector(-1,0){10}}\put(40,40){\vector(1,0){10}}
\put(30,50){\vector(1,0){10}}
\put(0,30){\vector(0,-1){10}}
\put(10,20){\vector(0,-1){10}}\put(10,20){\vector(0,1){10}}\put(10,37){\vector(0,-1){7}}
\put(20,10){\vector(0,-1){10}}\put(20,10){\vector(0,1){10}}\put(20,30){\vector(0,-1){10}}
 \put(20,30){\vector(0,1){10}}
\put(30,0){\vector(0,1){10}}\put(30,20){\vector(0,-1){10}}\put(30,20){\vector(0,1){10}}
 \put(30,40){\vector(0,-1){10}}\put(30,40){\vector(0,1){10}}
\put(40,10){\vector(0,1){10}}\put(40,30){\vector(0,-1){10}}\put(40,30){\vector(0,1){10}}
 \put(40,50){\vector(0,-1){10}}
\put(50,20){\vector(0,1){10}}\put(50,40){\vector(0,-1){10}}
{\linethickness{1.5pt}
\put(10,35){\line(0,-1){2}}
\put(13,30){\line(1,0){4}}
\put(20,27){\line(0,-1){4}}
\put(23,20){\line(1,0){4}}
\put(30,17){\line(0,-1){4}}
\put(33,10){\line(1,0){4}}
\put(40,13){\line(0,1){4}}
\put(37,20){\line(-1,0){4}}
\put(30,23){\line(0,1){4}}
\put(33,30){\line(1,0){4}}
\put(40,33){\line(0,1){4}}
\put(43,40){\line(1,0){4}}
\put(50,37){\line(0,-1){4}}
\put(53,30){\line(1,0){2}}
\put(10,30)\ne\put(20,30)\sw\put(20,20)\ne\put(30,20)\sw\put(30,10)\ne
\put(40,10)\nw\put(40,20)\sw\put(30,20)\ne\put(30,30)\se\put(40,30)\nw
\put(40,40)\se\put(50,40)\sw\put(50,30)\ne
}
\put(5,25)\bs\put(15,15)\bs\put(15,35)\bs\put(25,5)\bs\put(25,25)\bs\put(25,45)\bs
\put(35,15)\bs\put(35,35)\bs\put(45,25)\bs\put(45,45)\bs\put(55,35)\bs
\put(5,15)\ws\put(15,5)\ws\put(15,25)\ws\put(25,15)\ws\put(25,35)\ws\put(35,5)\ws
\put(35,25)\ws\put(35,45)\ws\put(45,15)\ws\put(45,35)\ws\put(55,25)\ws
\put(-5,25)\ws\put(5,35)\ws\put(25,-5)\ws\put(35,55)\bs
\put(4,21){\pa}
\put(14,11){\pa}\put(14,31){\pa}
\put(24,1){\pa}\put(24,21){\pa}\put(24,41){\pa}
\put(34,11){\pa}\put(34,31){\pa}\put(34,51){\pa}
\put(44,21){\pa}\put(44,41){\pa}
\put(54,31){\pa}
\put(4,31){\pb}
\put(14,21){\pb}
\put(24,11){\pb}\put(24,31){\pb}
\put(34,21){\pb}\put(34,41){\pb}
\put(44,31){\pb}
\put(0.5,24.5){\pd}
\put(10.5,14.5){\pd}\put(10.5,24.5){\pc}\put(10.5,34.5){\pd}
\put(20.5,4.5){\pd}\put(20.5,14.5){\pc}\put(20.5,24.5){\pd}\put(20.5,34.5){\pc}
\put(30.5,4.5){\pc}\put(30.5,14.5){\pd}\put(30.5,24.5){\pc}\put(30.5,34.5){\pd}\put(30.5,44.5){\pc}
\put(40.5,14.5){\pc}\put(40.5,24.5){\pd}\put(40.5,34.5){\pc}\put(40.5,44.5){\pd}
\put(50.5,24.5){\pc}\put(50.5,34.5){\pd}
\put(8,37){$a$}\put(57,27){$b$}
\put(10,35){\circle*{1}}
\put(55,30){\circle*{1}}
\end{picture}}
\caption{An example of a lattice domain $\Om$ with boundary conditions
creating an interface from $a$ to $b$, which is drawn in bold.
The lattice squares are colored in the chessboard fashion,
with black corresponding
to the sites of the original Ising lattice and white to the sites of its dual.
Near  the edges we write the arguments of the corresponding lines.
Note that running from $a$ to $b$ the interface always
follows the arrows and has black squares on the left.
}\label{fig:int}
\end{figure}

Denote $\HH$ restricted to black squares by $\HH^b$.
We define the {\em discrete Laplacian} by
$$ \Delta\HH^b(B):=\sum_j\br{\HH(B_j)-\HH(B)}~,$$
where the sum is taken over four black squares $B_j$ 
-- neighbors of the black square $B$, touching it at vertices. 
Similarly we define the Laplacian for the restriction $\HH^w$ to white squares.
We say that a function is {\em discrete (sub/super) harmonic} if its Laplacian vanishes 
(is non-negative/non-positive).

\begin{lem}\label{lem:harm}
If $F$ is preholomorphic, 
then for an interior white square $W$
with corner vertices $NW$, $NE$, $SE$, $SW$
(starting from the upper left and going clockwise) we have
\begin{equation}\label{eq:superharm}
\Delta\HH^w(W)~=~-\abs{F(NE)-F(SW)}^2~=-\abs{F(NW)-F(SE)}^2~\le0~,
\end{equation}
so $\HH^w$ is superharmonic.
Similarly, for an interior black square $B$
\begin{equation}\label{eq:subharm}
\Delta\HH^b(B)~=~\abs{F(NE)-F(SW)}^2~=\abs{F(NW)-F(SE)}^2~\ge0~,
\end{equation}
so $\HH^b$ is subharmonic.
\end{lem}

\begin{rem}
It is clear that subharmonicity on black squares
is equivalent to superharmonicity on white ones.
Indeed, the definition of $\HH$ is such that it is always increased when we pass
from white squares to black. If we reverse the colors, we will arrive at the function $-\HH$
and subharmonicity will become superharmonicity.
\end{rem}

\begin{proof}
A computation leading to equation (\ref{eq:superharm},\ref{eq:subharm}) is possible, since increment 
of $\HH$ across some vertex, say $NW$, can be written in terms of the
projections of $F(NW)$ 
on various lines, and so ultimately in terms of $F(NW)$.
So we can express the Laplacian $\Delta\HH(B)$
in terms of the values of $F$ at four neighboring vertices.
But projections of $F(NW)$ on two lines
corresponding to upper and left edges of the square
coincide with those of $F(NE)$ and $F(SW)$ correspondingly.
So we can express $F(NW)$ and similarly $F(SE)$
in terms of $F(NE)$ and $F(SW)$.
The resulting formula for the Laplacian is quite simple.

The computation is rather lengthy, so we present it in the Appendix~\ref{sec:subharm}.
However there are several arguments why we should arrive at a simple result.
Since we deal with squares of absolute values of projections,
we arrive at some real quadratic form in
$F(NE),\bar F(NE),F(SW),\bar F(SW)$.
Symmetries of our setup imply
that this form is invariant under the
rotation by $\pi$ which yields the exchange $F(NW)\leftrightarrow -F(SE)$,
and under change of $F$ by an additive constant
(this follows e.g. from the equations (\ref{eq:nw},\ref{eq:sw},\ref{eq:se},\ref{eq:sw})).
Such a form is necessarily given by
$\const\cdot\abs{F(NE)-F(SW)}^2$.
Note that any value of the constant would do, leading to a sub- or super- or harmonic
function.
\end{proof}

By {\em boundary arcs} we mean parts of $\partial\Om$
which are not separated by
``distinguished'' points
(i.e. ends of half-edges).
In our usual setup there are two boundary arcs,
$ab$ and $ba$
 (with points given counterclockwise).
By values of $\HH_F$ on the boundary
we mean its values on the outside squares adjacent to $\Om$.

\begin{lem}\label{lem:const}
If a preholomorphic function
$F$ solves the problem $\rbvp$,
then $\HH_F$ is constant
on the boundary arcs.
\end{lem}

\begin{proof}
Go along a boundary arc over the squares
adjacent to the domain.
Let $B$ and $B'$ be the centers of two consecutive ones 
(say of black color),
they touch at a vertex $v$, and are separated from $\Om$ by the
edges $e$ and $e'$ emanating from $v$.
Then
$$\HH(B)-\HH(B')=|F(e)|^2-|F(e')|^2=0~,$$
So $\HH$ is indeed constant along the arc.
\end{proof}

\section{Discrete holomorphic spin structure}
\label{sec:fermi}

We consider loop representation of
the random cluster Ising model at critical temperature.
The discrete domain $\Om$ on the lattice $\mesh\ZZ^2$
is as discussed above,
with a ``source'' $a$ and a ``sink'' $b$.
Thus for every configuration $\omega$ besides loops 
there is a curve $\ga=\ga(\omega)$ joining $a$ and $b$,
which we will call the {\em interface},
see Figure~\ref{fig:loops}.
Rotate the lattice in such a way so that an only edge incoming into $b$
from $\Om$ points to the right.

We round the corners of the loops
so that there are no sharp turns, see Figure~\ref{fig:int}.
The loops can be connected at a vertex in two different ways,
like near the vertex $v$ in Figure~\ref{fig:perestroika},
which are more clearly distinguished if we draw rounded loops.
Note that the interface can pass through a vertex twice,
utilizing two rounded corners -- see the left part of Figure~\ref{fig:perestroika}.
Color the lattice squares in a chessboard way,
so that standing at $b$ and facing the domain $\Om$,
we have a black square on the right and a white one on the left.
The black squares correspond to the sites of the original Ising lattice,
while the white squares correspond to the sites of the dual one.
An interface going from $a$ to $b$ always has 
black squares on the left and white on the right,
so it can arrive to a point $z$ only from one direction
(and not from the opposite one).
Thus we can prescribe to every point a vector
which is tangent to all interfaces passing through it
from $a$ to $b$.
Since it has black square on the right,
for points on the edge $e$ it coincides with the vector $\vec e$
discussed above.

Recall that to every point $z$ 
(centers of edges and rounded corners are important)
we prescribe a line $\ell(z)$ in the complex plane
which passes through the origin and the square root of the corresponding vector
(the choice of the square root is irrelevant).
For edges this agrees with the scheme discussed above,
see Figure~\ref{fig:int}.

A given vertex $v$ has 8 neighboring corner or edge centers.
If we go around $v$ counterclockwise, the line corresponding to the neighbor 
with each step is rotated clockwise by $\pi/8$ 
(so the full turn corresponds to a rotation by $\pi$,
which preserves the line but reverses directions).

For two points $z,z'$ on an interface $\ga$ we will denote by $\wi(z\to z')=\wi(\ga,z\to z')$
the winding (i.e. the total turn) of the curve $\ga$ as it goes from
$z$ to $z'$, measured in radians.

For an interface $\ga$ we define the {\em complex weight} $\we$
at point $z\in\ga$ by 
$$
\we(\ga,z)~:=~\exp\br{-\frac i4\br{\wi(\ga,a\to z)+\wi(\ga,b\to z)-\wi(\ga,a\to b)}}~.
$$
Note that $\wi(\ga,a\to z)-\wi(\ga,b\to z)=\wi(\ga,a\to b)$ and so
\begin{equation}\label{eq:defmarkov}
\we(\ga,z)~=~\exp\br{-\frac i2\,\wi(\ga,b\to z)}~.
\end{equation}
The values of the complex weight are illustrated in Figure~\ref{fig:weight}

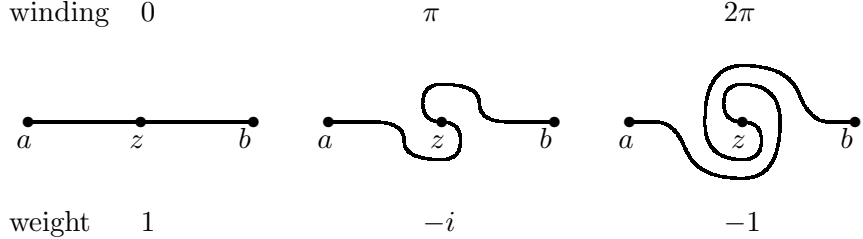
\begin{figure}\centering
\def\vtx{\circle*{3}}
\unitlength=0.5mm
\begin{picture}(220,70)
\def\green{}
\def\red{}
\def\blue{}
\def\black{}
\thicklines
\put(30,30){
\green\qbezier(-30,0)(0,0)(0,0)
\red\qbezier(30,0)(0,0)(0,0)
\black
\put(-30,0)\vtx
\put(30,0)\vtx
\put(0,0)\vtx
\put(-33,-7){$a$}
\put(26,-7){$b$}
\put(-3,-7){$z$}
\put(-35,27){\red winding}
\put(-0,27){\red $0$}
\put(-0,-29){\blue $1$}
\put(-35,-29){\blue  weight}
}
\put(110,30){
\red\qbezier(18,0)(30,0)(30,0)
\qbezier(10,5)(10,0)(18,0)
\qbezier(10,5)(10,10)(0,10)
\qbezier(0,10)(-5,10)(-5,5)
\qbezier(-5,5)(-5,0)(0,0)
\green\qbezier(-18,0)(-30,0)(-30,0)
\qbezier(-10,-5)(-10,0)(-18,0)
\qbezier(-10,-5)(-10,-10)(0,-10)
\qbezier(0,-10)(5,-10)(5,-5)
\qbezier(5,-5)(5,0)(0,0)
\black
\put(-30,0)\vtx
\put(30,0)\vtx
\put(0,0)\vtx
\put(-33,-7){$a$}
\put(26,-7){$b$}
\put(-3,-7){$z$}
\put(-5,27){\red $\pi$}
\put(-5,-29){\blue $-i$}
}
\put(190,30){
\green
\qbezier(-30,0)(-30,0)(-23,0)
\qbezier(-23,0)(-18,0)(-15,-7.5)
\qbezier(-15,-7.5)(-12,-15)(0,-15)
\qbezier(0,-15)(10,-15)(10,0)
\qbezier(10,0)(10,10)(0,10)
\qbezier(0,10)(-5,10)(-5,5)
\qbezier(-5,5)(-5,0)(0,0)
\red
\qbezier(30,0)(30,0)(23,0)
\qbezier(23,0)(18,0)(15,7.5)
\qbezier(15,7.5)(12,15)(0,15)
\qbezier(0,15)(-10,15)(-10,0)
\qbezier(-10,0)(-10,-10)(0,-10)
\qbezier(0,-10)(5,-10)(5,-5)
\qbezier(5,-5)(5,0)(0,0)
\black
\put(-30,0)\vtx
\put(30,0)\vtx
\put(0,0)\vtx
\put(-33,-7){$a$}
\put(26,-7){$b$}
\put(-3,-7){$z$}
\put(-5,27){\red $2\pi$}
\put(-5,-29){\blue $-1$}
}
\end{picture}
\caption{Values of the complex weight $\we\br{\ga,z}$
for different passages of the interface $\ga$ through $z$.
}\label{fig:weight}
\end{figure}

\begin{lem}\label{lem:phase}
For a  point $z$ and any realization of the interface $\ga$
the complex weight $\we(\ga,z)$ belongs to the line $\ell(z)$.
\end{lem}
\begin{proof}
When the interface is traced starting from $b$ 
the property is easily checked by induction.
At the center of the first edge the complex weight is equal to $1$,
and so belongs to the line through $1$.
When the interface turns by $\pm\theta$,
the winding $\wi(\ga,b\to z)$ is increased by $\pm\theta$,
whereas  $\wi(\ga,a\to z)$ is decreased by $\mp\theta$.
So the complex weight changes by a factor of 
$\exp\br{-i/4\cdot\br{-\mp\theta\pm\theta}}=\exp\br{\mp i\theta/2}$.
On the other hand, the line $\ell(e)$ is also
rotated by $\mp \theta/2$ since it passes through the complex conjugate of the square root
of the corresponding tangent vector which is rotated by $\pm \theta$. 
Here we use that the interface when traced from $b$  always has black squares on the right
and so goes in the opposite direction to the tangent vector.
\end{proof}

We will work with points $z$ which are either ``centers of corner turns''
(near every vertex there are 4 such points
-- see Figure~\ref{fig:perestroika}) or centers of edges.
For corner and edge points we can write
\begin{equation}\label{eq:wsi}
\we(\ga,z)~:=~\si^{\num(\ga,a\to z)+\num(\ga,b\to z)-\num(\ga,a\to b)}~
=~\si^{2\num(\ga,b\to z)}~.
\end{equation}
Here $\num(z\to z')=\num(\ga,z\to z')$ is
the number of $\pm\frac\pi2$ turns with sign the curve $\ga$
makes going from $z$ to $z'$ and
$\si=\exp(-i\pi/8)$.
Note that for corner points $\wi(a\to z)$ differs from $\num(a\to z)\cdot\pi/2$ by
$\pm\pi/4$ (the last half-turn before reaching $z$) but the difference enters
$\wi(a\to z)$ and $\wi(b\to z)$ with opposite signs and so cancels out.
 \begin{rem}
As was mentioned before,
the choice of weight is such that a relative weight of the interface
with an additional $2\pi$ twist around $z$ is $-1$.
Indeed, such a twist forces each of the two halves $a\to z$ and $b\to z$
to make four $\pi/2$-turns, so the weight
for one $\pi/2$-turn is $\si=\exp(-i\pi8)$,
which satisfies $\si^8=-1$.
\end{rem}
 \begin{rem}
Taking $\si=\exp(i\pi/8)$ instead,
one arrives at discrete anti-analytic functions.
\end{rem}

Define a function $F$ at all ``centers of corner turns''
(near every vertex there are 4 such points
-- see picture) by
$$F(c)~:=~
\Exp{\chi_{c\in\ga(\omega)}\cdot\we(\ga(\omega),c)}\cdot 2 \cos\frac\pi8~.$$ 
Similarly define $F$ for all centers of edges by
$$F(e)~:=~
\Exp{\chi_{e\in\ga(\omega)}\cdot\we(\ga(\omega),e)}~.$$ 
Different normalization arises because there are more corners than edges per vertex.
 \begin{rem}
The given definition of $F$ for edge centers works well only
for the square lattice at criticality 
(which is perhaps the most interesting case).
As an alternative one can start with our definition for corner centers,
and use the equation (\ref{eq:edgedef}) 
to define $F$ for edge centers.
This approach gives the same function in our setting,
but also generalizes to non-critical values of $p$ and to other lattices.
\end{rem}

With corners rounded, the interface
can go through a vertex $v$ in 4 different ways,
passing through one of the 4 neighboring corners $c_j$.
For an interior vertex $v$ we define $F$ as 
$$F(v)~:=~\sum_jF(c_j)/2~.$$
One can rephrase this as saying that
\begin{equation}\label{eq:fdef}
F(v)~=~
\Exp{\chi_{v\in\ga(\omega)}\cdot\we(\ga(\omega),v)}\cdot \cos\frac\pi8~,
\end{equation} 
where all passages of the interface through $v$ (there might be up to two)
are counted separately.

\begin{lem}%[North by Northwest]
\label{lem:proj}
For an interior vertex $v$ the values of $F$
at its $8$ neighbors are orthogonal projections  of $F(v)$
on $8$ corresponding lines.
%In particular, $F$ is preholomorphic in $\Om$.
\end{lem}

\begin{rem}
A boundary vertex $v$ has only two neighboring edges,
say $e$ and $e'$.
We define F(v) as a unique complex number which has orthogonal
projections $F(e)$ and $F(e')$ on the corresponding
lines $\ell(e)$ and $\ell(e')$.
It follows that $F$ is {\em preholomorphic} in $\Om$.
\end{rem}

\begin{rem}
The proof uses that we have a square lattice at $v$,
but the only global information needed is that the graph is planar.
So if we define $F$ on some planar graph which
has square lattice pieces, it will be preholomorphic there.
In a sequel we will discuss  generalizations
of preholomorphic functions to general planar graphs.
\end{rem}

\begin{proof}
When going around $v$ clockwise the line is rotated by $\pi/8$ with each step,
thus lines corresponding to antipodal neighbors are
at angle of $4\pi/8=\pi/2$, and so are orthogonal.
Hence the values of $F$ at two antipodes are orthogonal,
and are orthogonal projections of their sum on the corresponding lines.
So we will be proving in fact a stronger property, namely that
\begin{equation}\label{eq:sums}
F(NW)+F(SE)=F(NE)+F(SW)=F(W)+F(E)=F(N)+F(S)=F(v)~,
\end{equation}
where in each of the pairs two terms are orthogonal.
Here starting from the right and going clockwise
we denote $8$ neighbors of $v$ by 
$E$, $SE$, $S$, $SW$, $W$, $NW$, $N$, $NE$.

Recall that by definition
$$F(v)~=~\br{F(NE)+F(NW)+F(SW)+F(SE)}/2~,$$
so to establish the identity (\ref{eq:sums}) and the Lemma it is sufficient to show
that the sum of values of $F$ at two antipodal neighbors
is the same for 4 such pairs of antipodes.

Define an involution $\omega\mapsto\omega'$ on loop configurations,
which results from the rearrangement of connections at the point $v$.
For the random cluster formulation it corresponds to opening/closing 
the edge going through $v$.
To prove the linear identity it is sufficient to show that each
pair $\omega,\omega'$ of configurations
makes identical contributions to all 4 ``antipodal'' sums.
 
Consider some pair of configurations, say $\omega$ and $\omega'$. 
If the curve $\ga(\omega)$ does not pass through $v$, neither
does $\ga(\omega')$, and all contributions are zeroes.

Otherwise both curves pass through $v$.
Trace either of the curves from $a$ until it first arrives to the neighborhood of $v$.
Since it has black squares on the left it can arrive from 2 possible directions,
similarly when traced from $b$ can arrive from 2 other directions.

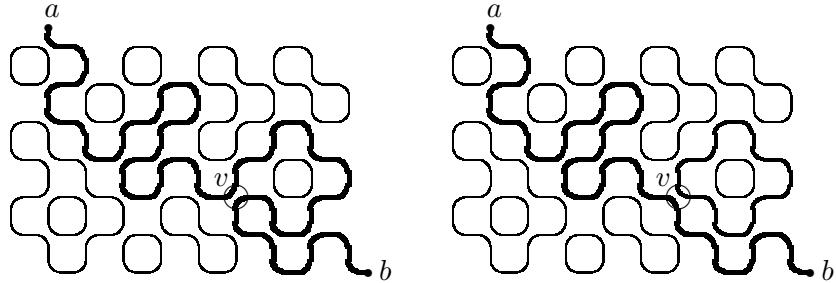
\begin{figure}\centering
{\def\bs{\circle*{3}}\def\ws{\circle{3}}
\def\ne{\qbezier(-0.,3)(0.5,0.5)(3,0)}
\def\nw{\qbezier(-0.,3)(-0.5,0.5)(-3,0)}
\def\se{\qbezier(-0.,-3)(0.5,-0.5)(3,0)}
\def\sw{\qbezier(-0.,-3)(-0.5,-0.5)(-3,0)}
\def\bne{\qbezier(-0.5,3)(0.,0.5)(3,-0)}
\def\bnw{\qbezier(-0.5,3)(-1,0.5)(-3,0)}
\def\bse{\qbezier(-0.5,-3)(0.,-0.5)(3,0)}
\def\bsw{\qbezier(-0.5,-3)(-1,-0.5)(-3,0)}
\def\uu{{\thicklines\line(1,1){10}}}
\def\dd{{\thicklines\line(1,-1){10}}}
\def\uw{\line(1,1){10}}
\def\dw{\line(1,-1){10}}
\unitlength=0.5mm
\centerline{
\begin{picture}(115,85)(-10,-10)
{\linethickness{1.5pt}
\put(10,65){\line(0,-1){2}}\put(10,47){\line(0,-1){4}}\put(20,37){\line(0,-1){4}}\put(20,57){\line(0,-1){4}}
\put(30,37){\line(0,-1){4}}\put(30,27){\line(0,-1){4}}\put(40,47){\line(0,-1){4}}\put(40,37){\line(0,-1){4}}
\put(40,27){\line(0,-1){4}}\put(50,47){\line(0,-1){4}}\put(50,27){\line(0,-1){4}}\put(60,27){\line(0,-1){4}}
\put(60,17){\line(0,-1){4}}\put(70,37){\line(0,-1){4}}\put(70,17){\line(0,-1){4}}\put(70,7){\line(0,-1){4}}
\put(80,37){\line(0,-1){4}}\put(80,17){\line(0,-1){4}}\put(80,7){\line(0,-1){4}}\put(90,27){\line(0,-1){4}}
\put(90,07){\line(0,-1){4}}\put(73,0){\line(1,0){4}}\put(63,10){\line(1,0){4}}\put(83,10){\line(1,0){4}}
\put(33,20){\line(1,0){4}}\put(53,20){\line(1,0){4}}\put(63,20){\line(1,0){4}}\put(83,20){\line(1,0){4}}
\put(93,00){\line(1,0){2}}\put(23,30){\line(1,0){4}}\put(33,30){\line(1,0){4}}\put(63,30){\line(1,0){4}}
\put(83,30){\line(1,0){4}}\put(13,40){\line(1,0){4}}\put(33,40){\line(1,0){4}}\put(43,40){\line(1,0){4}}
\put(73,40){\line(1,0){4}}\put(13,50){\line(1,0){4}}\put(43,50){\line(1,0){4}}\put(13,60){\line(1,0){4}}
\put(43,30){\line(1,0){4}}\put(73,10){\line(1,0){4}}
\put(10,60)\bne\put(10,40)\bne\put(20,30)\bne\put(30,20)\bne\put(50,20)\bne\put(60,10)\bne
\put(70,00)\bne\put(70,10)\bne\put(80,30)\bne
\put(20,50)\bnw\put(30,30)\bnw\put(40,20)\bnw\put(40,40)\bnw\put(50,40)\bnw\put(60,20)\bnw
\put(70,30)\bnw\put(80,00)\bnw\put(90,10)\bsw\put(90,20)\bnw\put(80,10)\bnw\put(40,30)\bnw
\put(20,60)\bsw\put(20,40)\bsw\put(50,30)\bsw\put(50,50)\bsw\put(70,10)\bsw\put(70,20)\bsw
\put(80,40)\bsw\put(90,30)\bsw
\put(10,50)\bse\put(30,30)\bse\put(30,40)\bse\put(40,30)\bse\put(40,40)\bse\put(40,50)\bse
\put(60,20)\bse\put(60,30)\bse\put(70,40)\bse\put(80,10)\bse\put(80,20)\bse\put(90,0)\bne
}
\put(10,00)\ne\put(30,00)\ne\put(50,00)\ne\put(00,10)\ne\put(40,10)\ne\put(50,10)\ne
\put(20,20)\ne\put(70,20)\ne\put(00,30)\ne\put(10,30)\ne\put(50,30)\ne\put(20,40)\ne
\put(80,40)\ne\put(00,50)\ne\put(30,50)\ne\put(50,50)\ne\put(60,50)\ne\put(70,50)\ne
\put(80,50)\ne\put(10,10)\ne
\put(20,00)\nw\put(40,00)\nw\put(60,00)\nw\put(20,10)\nw\put(30,10)\nw\put(10,20)\nw
\put(80,20)\nw\put(60,30)\nw\put(30,40)\nw\put(70,40)\nw\put(90,40)\nw\put(10,50)\nw
\put(40,50)\nw\put(60,40)\nw
\put(10,10)\sw\put(40,10)\sw\put(50,10)\sw\put(60,10)\sw\put(20,20)\sw\put(30,20)\sw
\put(50,20)\sw\put(10,30)\sw\put(20,30)\sw\put(80,30)\sw\put(10,40)\sw\put(60,50)\sw
\put(70,50)\sw\put(80,50)\sw\put(90,50)\sw\put(10,60)\sw\put(40,60)\sw\put(60,60)\sw
\put(80,60)\sw\put(30,50)\sw
\put(20,10)\se\put(30,10)\se\put(00,20)\se\put(10,20)\se\put(40,20)\se\put(70,30)\se
\put(00,40)\se\put(50,40)\se\put(20,50)\se\put(00,60)\se\put(30,60)\se\put(60,40)\se
\put(50,60)\se\put(70,60)\se
\put(00,17){\line(0,-1){4}}\put(00,37){\line(0,-1){4}}\put(00,57){\line(0,-1){4}}\put(10,07){\line(0,-1){4}}
\put(10,17){\line(0,-1){4}}\put(10,27){\line(0,-1){4}}\put(10,37){\line(0,-1){4}}\put(10,47){\line(0,-1){4}}
\put(20,07){\line(0,-1){4}}\put(20,17){\line(0,-1){4}}\put(20,27){\line(0,-1){4}}\put(20,47){\line(0,-1){4}}
\put(30,07){\line(0,-1){4}}\put(30,17){\line(0,-1){4}}\put(30,47){\line(0,-1){4}}\put(30,57){\line(0,-1){4}}
\put(40,07){\line(0,-1){4}}\put(40,17){\line(0,-1){4}}\put(40,57){\line(0,-1){4}}\put(50,07){\line(0,-1){4}}
\put(50,17){\line(0,-1){4}}\put(50,37){\line(0,-1){4}}\put(50,57){\line(0,-1){4}}\put(60,07){\line(0,-1){4}}
\put(60,37){\line(0,-1){4}}\put(60,57){\line(0,-1){4}}\put(60,47){\line(0,-1){4}}\put(70,27){\line(0,-1){4}}
\put(70,47){\line(0,-1){4}}\put(70,57){\line(0,-1){4}}\put(80,27){\line(0,-1){4}}\put(80,47){\line(0,-1){4}}
\put(80,57){\line(0,-1){4}}\put(90,47){\line(0,-1){4}}\put(10,57){\line(0,-1){4}}
\put(13,00){\line(1,0){4}}\put(33,00){\line(1,0){4}}\put(53,00){\line(1,0){4}}\put(03,10){\line(1,0){4}}
\put(13,10){\line(1,0){4}}\put(23,10){\line(1,0){4}}\put(33,10){\line(1,0){4}}\put(43,10){\line(1,0){4}}
\put(53,10){\line(1,0){4}}\put(03,20){\line(1,0){4}}\put(13,20){\line(1,0){4}}\put(23,20){\line(1,0){4}}
\put(43,20){\line(1,0){4}}\put(03,30){\line(1,0){4}}\put(13,30){\line(1,0){4}}\put(53,30){\line(1,0){4}}
\put(53,40){\line(1,0){4}}\put(63,40){\line(1,0){4}}\put(83,40){\line(1,0){4}}\put(03,50){\line(1,0){4}}
\put(03,40){\line(1,0){4}}\put(23,40){\line(1,0){4}}\put(73,30){\line(1,0){4}}\put(73,20){\line(1,0){4}}
\put(23,50){\line(1,0){4}}\put(33,50){\line(1,0){4}}\put(53,50){\line(1,0){4}}\put(63,50){\line(1,0){4}}
\put(73,50){\line(1,0){4}}\put(83,50){\line(1,0){4}}\put(03,60){\line(1,0){4}}\put(33,60){\line(1,0){4}}
\put(53,60){\line(1,0){4}}\put(73,60){\line(1,0){4}}
\put(09,68){$a$}\put(98,-2){$b$}\put(54,23){$v$}
\put(10,65){\circle*{2}}
\put(60,20){\circle{6}}
\put(95,00){\circle*{2}}
\end{picture}
\begin{picture}(115,85)(-10,-10)
{\linethickness{1.pt}
\put(70,37){\line(0,-1){4}}\put(60,27){\line(0,-1){4}}
\put(70,17){\line(0,-1){4}}\put(80,37){\line(0,-1){4}}
\put(80,17){\line(0,-1){4}}
\put(90,27){\line(0,-1){4}}
\put(73,40){\line(1,0){4}}
\put(63,20){\line(1,0){4}}\put(83,20){\line(1,0){4}}\put(63,30){\line(1,0){4}}
\put(83,30){\line(1,0){4}}\put(73,10){\line(1,0){4}}
\put(60,20)\bne\put(70,10)\bne\put(80,30)\bne
\put(70,30)\bnw\put(90,20)\bnw\put(80,10)\bnw
\put(70,20)\bsw\put(80,40)\bsw
\put(90,30)\bsw\put(80,20)\bse
\put(60,30)\bse\put(70,40)\bse}
{\linethickness{1.5pt}
\put(10,65){\line(0,-1){2}}\put(10,47){\line(0,-1){4}}\put(20,37){\line(0,-1){4}}\put(20,57){\line(0,-1){4}}
\put(30,37){\line(0,-1){4}}\put(30,27){\line(0,-1){4}}\put(40,47){\line(0,-1){4}}\put(40,37){\line(0,-1){4}}
\put(40,27){\line(0,-1){4}}\put(50,47){\line(0,-1){4}}\put(50,27){\line(0,-1){4}}\put(60,17){\line(0,-1){4}}
\put(70,7){\line(0,-1){4}}\put(80,7){\line(0,-1){4}}
\put(90,07){\line(0,-1){4}}\put(73,0){\line(1,0){4}}\put(63,10){\line(1,0){4}}\put(83,10){\line(1,0){4}}
\put(33,20){\line(1,0){4}}\put(53,20){\line(1,0){4}}\put(93,00){\line(1,0){2}}\put(23,30){\line(1,0){4}}
\put(33,30){\line(1,0){4}}\put(13,40){\line(1,0){4}}\put(33,40){\line(1,0){4}}\put(43,40){\line(1,0){4}}
\put(13,50){\line(1,0){4}}\put(43,50){\line(1,0){4}}\put(13,60){\line(1,0){4}}\put(43,30){\line(1,0){4}}
\put(10,60)\bne\put(10,40)\bne\put(20,30)\bne\put(30,20)\bne\put(50,20)\bne\put(60,10)\bne
\put(70,00)\bne\put(20,50)\bnw\put(30,30)\bnw\put(40,20)\bnw\put(40,40)\bnw\put(50,40)\bnw
\put(80,00)\bnw\put(90,10)\bsw\put(40,30)\bnw\put(20,60)\bsw\put(20,40)\bsw\put(50,30)\bsw
\put(50,50)\bsw\put(70,10)\bsw\put(10,50)\bse\put(30,30)\bse\put(30,40)\bse\put(40,30)\bse
\put(40,40)\bse\put(40,50)\bse
\put(60,20)\bsw\put(80,10)\bse\put(90,0)\bne
}
\put(10,00)\ne\put(30,00)\ne\put(50,00)\ne\put(00,10)\ne\put(40,10)\ne\put(50,10)\ne
\put(20,20)\ne\put(70,20)\ne\put(00,30)\ne\put(10,30)\ne\put(50,30)\ne\put(20,40)\ne
\put(80,40)\ne\put(00,50)\ne\put(30,50)\ne\put(50,50)\ne\put(60,50)\ne\put(70,50)\ne
\put(80,50)\ne\put(10,10)\ne
\put(20,00)\nw\put(40,00)\nw\put(60,00)\nw\put(20,10)\nw\put(30,10)\nw\put(10,20)\nw
\put(80,20)\nw\put(60,30)\nw\put(30,40)\nw\put(70,40)\nw\put(90,40)\nw\put(10,50)\nw
\put(40,50)\nw\put(60,40)\nw
\put(10,10)\sw\put(40,10)\sw\put(50,10)\sw\put(60,10)\sw\put(20,20)\sw\put(30,20)\sw
\put(50,20)\sw\put(10,30)\sw\put(20,30)\sw\put(80,30)\sw\put(10,40)\sw\put(60,50)\sw
\put(70,50)\sw\put(80,50)\sw\put(90,50)\sw\put(10,60)\sw\put(40,60)\sw\put(60,60)\sw
\put(80,60)\sw\put(30,50)\sw
\put(20,10)\se\put(30,10)\se\put(00,20)\se\put(10,20)\se\put(40,20)\se\put(70,30)\se
\put(00,40)\se\put(50,40)\se\put(20,50)\se\put(00,60)\se\put(30,60)\se\put(60,40)\se
\put(50,60)\se\put(70,60)\se
\put(00,17){\line(0,-1){4}}\put(00,37){\line(0,-1){4}}\put(00,57){\line(0,-1){4}}\put(10,07){\line(0,-1){4}}
\put(10,17){\line(0,-1){4}}\put(10,27){\line(0,-1){4}}\put(10,37){\line(0,-1){4}}\put(10,47){\line(0,-1){4}}
\put(20,07){\line(0,-1){4}}\put(20,17){\line(0,-1){4}}\put(20,27){\line(0,-1){4}}\put(20,47){\line(0,-1){4}}
\put(30,07){\line(0,-1){4}}\put(30,17){\line(0,-1){4}}\put(30,47){\line(0,-1){4}}\put(30,57){\line(0,-1){4}}
\put(40,07){\line(0,-1){4}}\put(40,17){\line(0,-1){4}}\put(40,57){\line(0,-1){4}}\put(50,07){\line(0,-1){4}}
\put(50,17){\line(0,-1){4}}\put(50,37){\line(0,-1){4}}\put(50,57){\line(0,-1){4}}\put(60,07){\line(0,-1){4}}
\put(60,37){\line(0,-1){4}}\put(60,57){\line(0,-1){4}}\put(60,47){\line(0,-1){4}}\put(70,27){\line(0,-1){4}}
\put(70,47){\line(0,-1){4}}\put(70,57){\line(0,-1){4}}\put(80,27){\line(0,-1){4}}\put(80,47){\line(0,-1){4}}
\put(80,57){\line(0,-1){4}}\put(90,47){\line(0,-1){4}}\put(10,57){\line(0,-1){4}}
\put(13,00){\line(1,0){4}}\put(33,00){\line(1,0){4}}\put(53,00){\line(1,0){4}}\put(03,10){\line(1,0){4}}
\put(13,10){\line(1,0){4}}\put(23,10){\line(1,0){4}}\put(33,10){\line(1,0){4}}\put(43,10){\line(1,0){4}}
\put(53,10){\line(1,0){4}}\put(03,20){\line(1,0){4}}\put(13,20){\line(1,0){4}}\put(23,20){\line(1,0){4}}
\put(43,20){\line(1,0){4}}\put(03,30){\line(1,0){4}}\put(13,30){\line(1,0){4}}\put(53,30){\line(1,0){4}}
\put(53,40){\line(1,0){4}}\put(63,40){\line(1,0){4}}\put(83,40){\line(1,0){4}}\put(03,50){\line(1,0){4}}
\put(03,40){\line(1,0){4}}\put(23,40){\line(1,0){4}}\put(73,30){\line(1,0){4}}\put(73,20){\line(1,0){4}}
\put(23,50){\line(1,0){4}}\put(33,50){\line(1,0){4}}\put(53,50){\line(1,0){4}}\put(63,50){\line(1,0){4}}
\put(73,50){\line(1,0){4}}\put(83,50){\line(1,0){4}}\put(03,60){\line(1,0){4}}\put(33,60){\line(1,0){4}}
\put(53,60){\line(1,0){4}}\put(73,60){\line(1,0){4}}
\put(09,68){$a$}\put(98,-2){$b$}\put(54,23){$v$}
\put(10,65){\circle*{2}}
\put(60,20){\circle{6}}
\put(95,00){\circle*{2}}
\end{picture}}
}
\caption{Rearrangement at a vertex $v$: we only change connections inside
 a small circle marking $v$.
Either interface does not pass through $v$ in both configurations,
or it passes in ways similar to the pictured above.
On the left the interface (in bold) passes through $v$ twice,
on the right (after the rearrangement) it passes once, but a new loop through $v$
appears.
The loops not passing through $v$ remain the same,
so the weights of configurations
differ by a factor of $\sqrt{q}=\sqrt{2}$ because of 
the additional loop on the right.
To get some linear relation on values of $F$,
it is enough to check that any pair of such configurations
makes equal contributions to two sides of the relation.}
\label{fig:perestroika1}
\end{figure}

\begin{figure}{
\def\pa{$0$}\def\pb{$\frac\pi2$}\def\pc{$\frac{3\pi}4$}\def\pd{$\frac\pi4$}
\def\ne{\qbezier(0,3)(0.5,0.5)(3,0)}
\def\nw{\qbezier(0,3)(-0.5,0.5)(-3,0)}
\def\se{\qbezier(0,-3)(0.5,-0.5)(3,0)}
\def\sw{\qbezier(0,-3)(-0.5,-0.5)(-3,0)}
\unitlength=1.5mm
\begin{picture}(100,50)
\put(5,43){configuration~$\omega'$}
\put(17,20){\vector(-1,0){17}}
\put(20,17){\vector(0,-1){17}}
\put(23,20){\line(1,0){14}}
\put(20,23){\line(0,1){14}}
\put(23,40){\line(1,0){14}}
\put(40,23){\line(0,1){14}}
\put(20,40)\se\put(40,40)\sw\put(40,20)\nw

{\put(20,20)\se\put(20,20)\nw}
\put(21.1,21.1){\circle*{1}}\put(21.5,22){$NE$}
\put(21.1,18.9){\circle*{1}}\put(21.5,17){$SE$}
\put(18.9,21.1){\circle*{1}}\put(13.5,22){$NW$}
\put(18.9,18.9){\circle*{1}}\put(13.5,17){$SW$}
\put(20,30){\circle*{1}}\put(21,30){$N$}
\put(30,20){\circle*{1}}\put(30,21){$E$}
\put(10,20){\circle*{1}}\put(7,21){$W$}
\put(20,10){\circle*{1}}\put(21,10){$S$}
\put(0,21){to~$a$}
\put(21,0){to~$b$}
\put(50,0){\put(5,43){configuration~$\omega$}
\put(17,20){\vector(-1,0){17}}
\put(20,17){\vector(0,-1){17}}
\put(23,20){\line(1,0){14}}
\put(20,23){\line(0,1){14}}
\put(23,40){\line(1,0){14}}
\put(40,23){\line(0,1){14}}
\put(20,40)\se\put(40,40)\sw\put(40,20)\nw
{\put(20,20)\sw\put(20,20)\ne}
\put(21.1,21.1){\circle*{1}}\put(21.5,22){$NE$}
\put(21.1,18.9){\circle*{1}}\put(21.5,17){$SE$}
\put(18.9,21.1){\circle*{1}}\put(13.5,22){$NW$}
\put(18.9,18.9){\circle*{1}}\put(13.5,17){$SW$}
\put(20,30){\circle*{1}}\put(21,30){$N$}
\put(30,20){\circle*{1}}\put(30,21){$E$}
\put(10,20){\circle*{1}}\put(7,21){$W$}
\put(20,10){\circle*{1}}\put(21,10){$S$}
\put(0,21){to~$a$}
\put(21,0){to~$b$}
}
\end{picture}}
\caption{Schematic drawing, representing the 
change in the interface after the rearrangement at a vertex $v$. 
Going from $a$ and $b$ to $v$ the interface might make some number of turns,
which won't influence our reasoning, since it changes the weight
of both configurations by the same factor. 
Note that since $a$ and $b$ are on the boundary, for topological
reasons the interface can go from $N$ to $E$
only on one side of $v$.
}\label{fig:perestroika}
\end{figure}
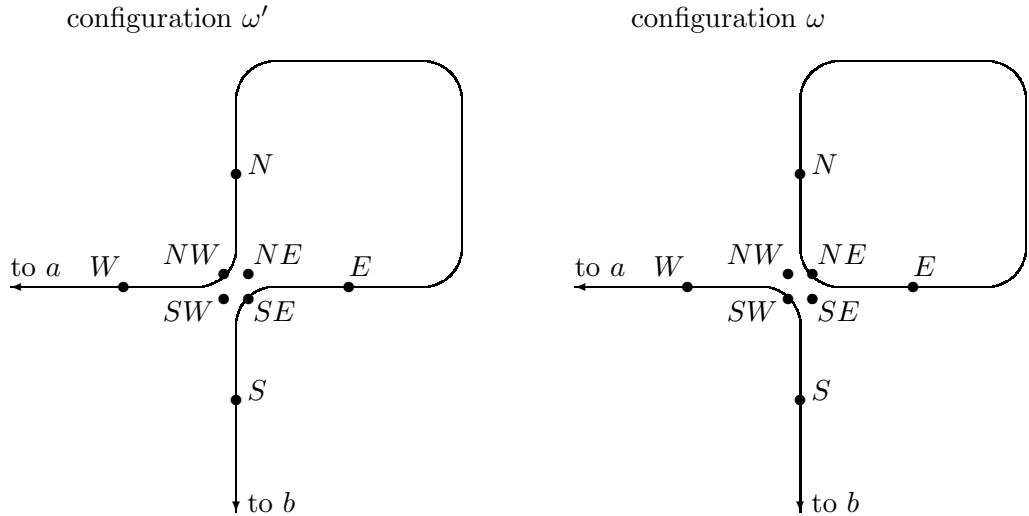

There are 4 possible topological pictures for the arrivals, but they are all analogous,
so we will work out one of them.
Assume that the half starting from $a$ arrives from the west,
while the half starting from $b$ arrives from the south
(such picture is possible for a half of the vertices $v$,
for others the curve traced from $a$ would arrive from a vertical direction).
In one of the curves, say $\ga(\omega)$, the two traced halves are immediately joined 
(and there is also a cycle passing near $v$), whereas in the other, $\ga(\omega')$,
this cycle is included into the curve.
See Figures~\ref{fig:perestroika1} and \ref{fig:perestroika}.
%    _        N
%   | |     NW|NE
% a-z-     W--z--E
%   |       SW|SE
%   b         S
Then out of corner points $\ga(\omega)$ contributes only to $F(SW)$,
say a term $X$ (weight of all cycles times the complex weight).
On the other hand the curve $\ga(\omega')$ out of the corner points contributes to
$F(NW)$ and $F(SE)$ only.
The contributions differ from $X$ by a factor of $1/\sqrt{2}$,
since the number of cycles decreased by one. 
Moreover, the phase changes, since compared to $\ga(\omega)$ reaching $SW$
the curve $\ga(\omega')$ winds by  additional $+\pi/2$ when reaching $NW$
(coming from the half originating in $b$)
and by additional $-\pi/2$ when reaching $SE$
(coming from the half originating in $a$).
Correspondingly the complex weights change
by factors of $\si^2$ and $\bar\si^2$, see equation (\ref{eq:wsi}).
So $\ga(\omega')$ contributes to
$F(NW)+F(SE)$ a term
$$X\cdot(\si^2+\bar\si^2)/\sqrt{2}
~=~X\cdot\br{e^{-i\frac\pi8 2}+e^{i\frac\pi8 2}}/\sqrt{2}
~=~X\cdot2\cos\br{\frac\pi4}/\sqrt{2}~=~X~.$$
So $\ga(\omega)$
contributes $X$ to the second sum in (\ref{eq:sums}),
while not contributing to the first,
whereas $\ga(\omega')$
contributes $X$ to the first sum in (\ref{eq:sums}),
while not contributing to the second.
We conclude that the first two sums coincide:
$$F(NW)+F(SE)~=~F(NE)+F(SW)~.$$

The same (but messier) reasoning shows that the two remaining sums share the same value.
Perhaps this is best summarized by the following table which shows
contributions of two configurations to the values of $F$ at 
various neighbors of $v$:

\medskip
\noindent
\begin{tabular}{||c||c|c||c|c||c|c||c|c||}
\hline &NW&SE&NE&SW&N&S&W&E\\ \hline\hline
$\displaystyle\omega$& 0& 0& 0& $\displaystyle X$ & 0&
$\displaystyle\frac{X\bar\si}{2\cos(\pi/8)}$& $\displaystyle\frac{X\si}{2\cos(\pi/8)}$& 0\\ \hline
$\displaystyle \omega'$& $\displaystyle\frac{X\si^2}{\sqrt2}$& 
$\displaystyle\frac{X\bar\si^2}{\sqrt2}$ & 0& 0& 
$\displaystyle\frac{X\si^3}{2\cos(\pi/8)\sqrt{2}}$& 
$\displaystyle\frac{X\bar\si}{2\cos(\pi/8)\sqrt{2}}$& $\displaystyle\frac{X\si}{2\cos(\pi/8)\sqrt{2}}$&
$\displaystyle\frac{X\bar\si^3}{2\cos(\pi/8)\sqrt{2}}$\\
\hline  
\end{tabular}

\bigskip

Using that $\si=\exp\br{-i\pi/8}$, an exercise in trigonometry
one checks that numbers in $2\times2$ squares bordered by double lines
always sum up to $X$. Thus taken together $\omega$ and $\omega'$ make
identical contributions to all 4 antipodal sums in (\ref{eq:sums}).
 
Alternatively we can finish the proof 
by deducing that values of $F$ 
on the neighboring edges are also projections of $F(v)$.
To that effect we write the value of $F$ at the northern edge 
in terms of northwest and northeast corners.
Consider some edge $e$ emanating from $v$ with the corresponding line $\ell(e)$
passing through a vector $\alpha$.
Let $c$ and $c'$ be the two adjacent corner points. 
The corresponding lines are rotations of $\ell(e)$ by $\pm\pi/8$,
passing through vectors $\si\alpha$ and $\bar\si\alpha$ correspondingly.

Note that the interface passes through $e$
if and only if it passes through exactly one of
the points $c$ and $c'$.
Taking into account the difference in complex weight and normalization,
and using the formula (\ref{eq:proj})
for projections, we write
\begin{align}
F(e)&=\br{\bar\si F(c)+\si F(c')}/\br{2\cos(\pi/8)}\label{eq:edgedef}\\
&=\br{\bar\si\Proj{F(v),\si\alpha}+\si \Proj{F(v),\bar\si\alpha}}/\br{2\cos(\pi/8)}\nonumber\\
&=\br{\bar\si\br{F(v)+(\si\alpha)^2\bar F(v)}+\si \br{F(v)+(\bar\si\alpha)^2\bar F(v)}}/\br{4\cos(\pi/8)}\nonumber\\
&=\br{\bar\si F(v)+\si \alpha^2 \bar F(v) +\si  F(v)+\bar\si \alpha^2 \bar F(v)}/\br{4\cos(\pi/8)}\nonumber\\
&=\br{F(v)+\alpha^2 \bar F(v)}\br{\bar\si+\si}/\br{2(\bar\si+\si)}=\Proj{F(v),\alpha}~,\nonumber
\end{align}
thus finishing the proof.
\end{proof}

\begin{lem}\label{lem:cont}
For every positive $r$ there is a function $\delta_r(x):\,\RR_+\to\RR_+$ 
such that $\lim_{x\to0}\delta_r(x)=0$ and if two neighboring squares $B$ and $W$
are $r$ away from at least one of the boundary arcs $ab$ or $ba$, then
$$\abs{\HH(B)-\HH(W)}~\le~\delta_r(\mesh)~.$$
\end{lem}

\begin{rem}
Here $\mesh$ is the lattice step.
Note that the only way the shape of $\Om$
enters into the estimate is via $r$.
The Lemma essentially means that the restrictions of the function $\HH$
to black and white squares are uniformly close to each other when we are away from $a$ and $b$.
\end{rem}
 \begin{rem}\label{rem:cont}
The Lemma is derived from the fact that $F\to0$ away from $a$ and $b$ as $\mesh\to0$.
Since preholomorphic $F$ is uniquely determined by its boundary conditions,
there should be an Ising-independent proof,
using only discrete analyticity and boundary conditions. 
Unfortunately, we were not able to find a simple one.
\end{rem}

\begin{proof}
If an edge $e$ separates the squares $B$ and $W$,
by definition
\begin{equation}\label{eq:2arm}
\abs{\HH(B)-\HH(W)}=\abs{F(e)}^2\le\Prob{e\in\gamma}^2,
\end{equation}
so we can take as our function $\delta$ the square of the 
similar function in Lemma~\ref{lem:notspacefilling}.
\end{proof}.

\begin{lem}\label{lem:rbvp}
The function $F$ satisfies the Riemann Boundary Value Problem $\rbvp$.
Moreover $\HH=0$ on the (counterclockwise) boundary arc $ab$ 
and $\HH=1$ on the (counterclockwise) boundary arc $ba$. 
\end{lem}

\begin{proof}
Let $v$ be a boundary vertex with incoming edges $e$ and $e'$.
Then all possible interfaces $\gamma$ arrive at $e$ from $a$
(or from $b$) with the same winding, so
$\we(\gamma(\omega),e)$ is independent of $\omega$.
Therefore 
$$\abs{F(e)}=\abs{\Exp{\chi_{e\in\ga(\omega)}\cdot\we(\ga(\omega),e)}}
=\abs{\we(\ga(\omega),e)\,\Exp{\chi_{e\in\ga(\omega)}}}=\Prob{e\in\ga(\omega)}~.$$
Similarly
$\abs{F(e')}=\Prob{e'\in\ga(\omega)}$.
But since there are only two edges out of $v$,
an interface passes through $e$ if and only if it passes through $e'$, so
$$\abs{F(e)}=\Prob{e\in\ga(\omega)}=\Prob{e'\in\ga(\omega)}=\abs{F(e')}~,$$
and $F$ satisfies the Riemann Boundary Value Problem $\rbvp$,
proving the first statement of the Lemma.

By Lemma~\ref{lem:const} it follows that the function $\HH$ is 
constant on the boundary arcs.
Let $u$ and $v$ be the centers of squares
immediately below and above $b$.
Recall that we chose $\HH$ (which is determined up to an additive constant)
so that $\HH(u)=0$.
Thus  $\HH=0$ on the (counterclockwise) boundary arc $ab$.
Every interface passes through $b$, and furthermore has the same complex weight at $b$.
So 
$$\HH(v)=\HH(u)+\abs{F(b)}^2=0+\abs{\Exp{\chi_{b\in\ga}\cdot\we(\ga,b)}}^2
=\abs{\we(\ga,b)\,\Exp{\chi_{b\in\ga}}}^2=\Prob{b\in\ga(\omega)}^2=1~.$$
Therefore $\HH=1$ on the (counterclockwise) boundary arc $ba$. 
\end{proof}

When establishing the \slee~ connection in the sequel \cite{smirnov-fk2}, we will need the ``martingale property''
of $F$ with respect to the interface.
Consider the interface $\gamma$ as a random curve 
drawn from $a$ to $b$ with some parameterization.
Let $t<s$ be two stopping times
(actually we do not need a full parameterization at this point).
For the time $t$ (and similarly for $s$)
 we denote by $\gamma(t)$ the corresponding curve point,
and by $\gamma[0,t]$ the part of the curve from $\gamma(0):=a$ to $\gamma(t)$.
When speaking of domain $\Om\setminus\gamma[0,t]$, we
will actually mean its component of connectivity containing $b$.

Then the following holds
\begin{lem}\label{lem:discretets}
Let $z$ be a lattice vertex such that for every realization of the interface 
$z\in\Om\setminus\gamma[0,s]$.
Then for every realization of $\gamma[0,t]$
\begin{equation}
F(z,\Omega\setminus\gamma[0,t],\gamma(t),b)=
\Exps{\gamma[t,s]}{
F(z,\Omega\setminus\gamma[0,s],\gamma(s),b)|\gamma[0,t]}.
\nonumber 
\end{equation} 
The conditional expectation above is taken over all possible
continuations $\gamma[t,s]$ of the interface until the time $s$
assuming the part $\gamma[0,t]$ is given.
\end{lem}

\begin{proof}
With a fixed $\gamma[0,t]$ consider the remaining part $\gamma'$ of the curve $\gamma$
in the domain $\Om\setminus\gamma[0,t]$.
Plugging into the definition (\ref{eq:fdef}) of $F$ in this domain 
the formula (\ref{eq:defmarkov}) for the complex weight we obtain
$$F(z,\Om\setminus\gamma[0,t],\gamma(t),b)~=~
\Exp{X}~,$$
where the random variable $X$ is given by
$$X~:=~{\chi_{z\in\ga'}\cdot\exp\br{-\frac i2\,\wi(\ga',b\to z)}}\cdot \cos\frac\pi8~.
$$
Now the desired formula is the law of total expectation applied to $X$ and
the curve $\gamma[t,s]$:
$$\Exp{X}=\Exps{\gamma[t,s]}{\Exp{X|\gamma[t,s]}}~.$$
\end{proof}

\section{Passing to a limit}\label{sec:limit}

In this Section we prove the Theorem~\ref{thm:fermi}.
To derive convergence we will use only discrete analyticity 
and boundary values of $F$, and appeal directly to the Ising connection only
in quoting Lemma~\ref{lem:cont}.
As discussed in Remark~\ref{rem:cont},
the latter should have an Ising-independent proof.
So essentially the Theorem~\ref{thm:fermi} can be restated as a Theorem about preholomorphic functions
solving the Riemann boundary value problem (\ref{eq:rbvp}).

We work with a sequence of lattice domains,
which approximate 
a given domain $\Om$.
Consider a sequence of lattice domains $\Om_j$ 
with distinguished points $a_j$, $b_j$ and
with lattice steps $\mesh_j$.
Let  $F_j=F(z,\Om_j,a_j,b_j,\mesh_j\ZZ^2)$ be the expectation
as defined above and
denote $\HH_j:=\HH_{F_j}$.

Assume that
$\mesh_j\to0$ and $\tilde{\Om}_j,a_j,b_j\cara\Om,a,b$
as $j\to\infty$.
We use Carath\'eodory convergence of domains,
which is the convergence of normalized Riemann uniformization maps on compact subsets.

Recall that $\Phi$ is a mapping of $\Om$ to a horizontal strip $\RR\times[0,2]$,
such that $a$ and $b$ are mapped to $\mp\infty$.
Note that since $\Phi$ is uniquely defined up to a real additive constant,
its derivative, and hence the right hand side in (\ref{eq:convf}) are uniquely determined.
Recall Remark~\ref{rem:cara} that solution to the continuum Riemann boundary value problem
\rbvp, the function $\sqrt{\Phi'}$, is Carath\'eodory-stable.

\begin{rem}
The convergence (\ref{eq:convf}) in the Theorem~\ref{thm:fermi} holds on the boundary of $\Omega$ 
wherever it is a horizontal or vertical segment.
Since complex weight on the boundary is independent of configuration, we conclude that
for such $z\in\partial\Omega$
$$\frac1{\sqrt{\mesh_j}}\,{\Prob{z\in\gamma}}~\unif~ f(z):= \sqrt{\Phi'(z)}~,$$
from which one deduces that random cluster intersects the (smooth) boundary
on a set of dimension $1/2$ and that for the spin Ising model at criticality
 magnetization on the boundary is proportional to $\sqrt{\mesh}$
(with a specific factor).
\end{rem}

We start by establishing convergence of $\HH$'s:
\begin{lem}\label{lem:convh}
Away from $a$ and $b$ there is a uniform convergence:
\begin{equation}
\HH_{F_j}~\unif~h:=\IM \Phi\label{eq:convh}~.
\end{equation}
\end{lem}

\begin{proof}
Remove the union $V$ of some neighborhoods of $a$ and $b$, then there is a positive $r$
such that remaining parts of the boundary arcs $ab$ and $ba$
are at least $4r$-apart.
Then all points in $\Om\setminus V$ are at least $2r$ away
from at least one of the arcs $ab$ and $ba$.
We conclude that because of the 
Carath\'eodory convergence,
for sufficiently large $j$ all points in
in $\Om_j\setminus V$ are at least $r$ away
from at least one of the (discrete) arcs $ab$ and $ba$.

Hence by Lemma~\ref{lem:cont}
we have uniform convergence
\begin{equation}\label{eq:bwnclose}
\sup\abs{\HH(B_j)-\HH(W_j)}~=:~\delta_j~
\stackrel{j\to\infty}{\longrightarrow}0~,
\end{equation}
for neighboring squares $B_j,W_j\in\Om_j\setminus V$.

Considering $\HH_j$ restricted to black and white squares
we obtain functions 
$\HH_j^b$ and $\HH_j^w$
(superharmonic and subharmonic correspondingly by Lemma~\ref{lem:harm}).
If we extend $\HH_j^b$  from black ($\HH_j^w$ from white) 
lattice vertices to whole of $\Om$
in any reasonable way (e.g. making constant on lattice squares),
then  (\ref{eq:bwnclose}) can be rewritten
as convergence in the uniform norm on $\Om\setminus V$:
\begin{equation}
\norm{\HH_j^b-\HH_j^w}_{\Om\setminus V,\infty}~=:~\delta_j~
\stackrel{j\to\infty}{\longrightarrow}0~.\label{eq:bwclose}
\end{equation}

Let $\tilde\HH_j^b$ be a discrete harmonic function
on black squares with boundary values given by $h$,
define $\tilde\HH_j^w$ similarly.
Then
\begin{equation}
\HH_j^b+2\delta_j>\tilde\HH_j^b+\delta_j>\tilde\HH_j^w-\delta_j>\HH_j^w-2\delta_j
\end{equation}
on the boundary,
and hence inside domain $\Om$ since the four functions involved
are superharmonic, harmonic, harmonic and subharmonic correspondingly.
Together with (\ref{eq:bwclose}) this means that 
$$
\norm{\HH_j^b-\tilde\HH_j^b}_{\Om\setminus V,\infty}~<~5\,\delta_j~,
$$
and since $\tilde H_j$'s converge to $h$ by the Lemma~\ref{lem:cfl1},
so do $H_j$'s.
\end{proof}

If $H_j$'s are harmonic, the Theorem immediately follows.
Indeed, derivatives of $H_j$ admit an integral representation (in terms of $H_j$ itself),
so uniform convergence of $H_j$ implies uniform convergence inside
$\Omega$ of $\nabla H_j$ and hence its square root, i.e. $F_j$.
For a general approximately harmonic $H_j$ this doesn't work,
but in our case we can use that appropriate restriction of $H_j$ is subharmonic
and that  $F_j$ is exactly preholomorphic.

First we will need the following compactness estimate:

\begin{lem}\label{lem:fl2}
Let $U$ be a subdomain compactly contained in $\Omega$,
and denote by $U_j$ its discretizations with mesh $\mesh_j$.
Then 
 $${\mesh_j}\sum_{U_j}\abs{F_j}^2$$
is uniformly bounded.
\end{lem}
\begin{rem}
Note that the expression above is essentially $L^2$ norm of $F_j/{\sqrt{\mesh_j}}$.
\end{rem}

\begin{proof}
Note that when we jump diagonally over a vertex $z$,
the function $H_j$ changes by $\RE F_j^2$ or $\IM F_j^2$
depending on the direction.
It follows that it is enough to prove uniform boundedness of
 $${\mesh_j}\sum_{U_j}\abs{\nabla H_j}$$
where $\nabla$ denotes discrete difference gradient of $H_j$
restricted to black or white vertices.

From now on, we will work with $H_j$ on the ``black'' sublattice,
having centers of black squares as vertices.
In particular, $\Delta$ will denote the corresponding Laplacian.
Recall that restriction of $H_j$ to this lattice is subharmonic, i.e. $\Delta H_j\ge0$.

Fix a  square $Q$ of side length $l$  such that a nine times bigger square $9Q$
is contained in  $\Omega$.
We will denote by the same letter the
lattice approximation $Q=Q_j$ of $Q$,
whose side will have $L=L_j:=l/\mesh_j$ edges.
It is enough to find for every such $Q$ a uniform in $j$
bound on
\begin{equation}\label{eq:qbound}
{\mesh_j}\sum_{Q_j}\abs{\nabla\HH_j}.
\end{equation}

\def\HHH{\HH^*}
Recall that on $9Q$ one has
uniform convergence $\HH_j\unif \HH$.
Denote by $\HHH_j$ the discrete harmonic function on $9Q_j$
having the same boundary values as $\HH_j$.
Then $\HHH_j\unif\HH$ on $\partial 9Q$,
and by the Lemma~\ref{lem:cfl}
$\frac1{\mesh_j}\nabla\HHH_j\unif\nabla\HH$ on $Q$.
Summing over $Q$ we infer that
$${\mesh_j}\sum_{Q_j}\abs{\nabla\HHH_j}$$
converges to the area integral of $\abs{\nabla\HH}$ and in particular is bounded.
Thus to bound (\ref{eq:qbound}) it is enough to bound 
$${\mesh_j}\sum_{Q_j}\abs{\nabla\br{\HH_j-\HHH_j}}.$$

Since $\HH_j$ and $\HHH_j$ have the same limit, their difference converges
uniformly to zero:
\begin{equation}\label{eq:convdif}
\sup_{9Q}\abs{\HHH_j-\HH_j}=o_j\to0~~~\text{when}~~j\to\infty.
\end{equation}

Denote by $G_j(x,y)=G_{9Q_j}(x,y)$ the discrete Green's function for the square $9Q_j$,
with $\Delta_x G(x,y)=\Delta_y G(x,y)=\delta_{x-y}$ and vanishing on the boundary of $9Q\times9Q$.
Note that it is negative inside $9Q$.

Using discrete analogue of the Riesz representation formula
we can write a subharmonic function $\HH_j-\HHH_j$
as a convolution of its Laplacian
(which coincides with that of $\HH_j$) with the Green's function:
\begin{equation}
\label{eq:hjriesz}
\HH_j(x)-\HHH_j(x)=\sum_{y\in 9Q} \Delta \HH_j(y) G(x,y).
\end{equation}
Taking difference gradient, we arrive at
\begin{equation}
\label{eq:nablahjriesz}
\nabla\br{\HH_j(x)-\HHH_j(x)}=\sum_{y\in 9Q} \Delta \HH_j(y) \nabla_x G(x,y) .
\end{equation}

Using the Lemma~\ref{lem:green} we write:
\begin{align}
\sum_{x\in Q}\abs{\nabla\br{\HH_j(x)-\HHH_j(x)}}&=
\sum_{x\in Q}\abs{\sum_{y\in 9Q} \Delta \HH_j(y) \nabla_x G(x,y)}
\le\sum_{y\in 9Q} {\Delta \HH_j(y)}\sum_{x\in Q} \abs{\nabla_x G(x,y)}\nonumber\\
&\le\frac{\const}{L}\sum_{y\in 9Q} {\Delta \HH_j(y)}\sum_{x\in Q} \abs{G(x,y)}
=\frac{\const}{L}\sum_{x\in Q}\sum_{y\in 9Q} {\Delta \HH_j(y) G(x,y)}\nonumber\\
&=\frac{\const}{L}\sum_{x\in Q}\abs{{\HH_j(x)-\HHH_j(x)}}< 
\const \, L\, o_j=\frac{\const\, l\,o_j}{\mesh_j}~,\nonumber
\end{align}
proving the Lemma.
\end{proof}

Lemma~\ref{lem:fl2} implies 
(by Theorem V.12.a in \cite{\ferrandbook}
applied to the primitives of $F_j$'s )
 that the sequence $\{F_j/\sqrt{\mesh_j}\}$
is precompact in the uniform topology on any compact subset of $\Omega$.
To show its convergence to $f$ it is sufficient
to establish convergence to $f$ of a uniformly converging (say to $g$) subsequence.
Uniform convergence implies convergence of the (discrete) integral of $F_j^2/{\mesh_j}$ to that of $g^2$.
The imaginary part of the former is given by $\HH_j+\const$,
which also converges to $\IM\Phi$,
so we conclude that the two limits are equal:
$$\IM\Phi=\lim_j\HH_j=\IM\int g^2+\const.$$
Since the function involved are analytic, equality of imaginary parts
implies that they are equal up to a constant.
Differentiating and taking the square root, we conclude
that $g=\sqrt{\Phi'}$, thus proving the Theorem~\ref{thm:fermi}.

\appendix
\section{A priori estimates}\label{sec:apriori}

We use an estimate on the modulus of continuity of our function $F$,
which essentially states that the interface cannot be  space filling.
It reduces to 
rather weak (compared to what is known)
magnetization estimates,
which ultimately can be retrieved from the (old) literature.
However it seems that ``modern'' proofs are elusive and would be worthwhile.
It is also possible to build everything on the basis
of discrete analyticity, without appealing to properties of the Ising model.
But for now we present a sketch of the proof using classical results
and assuming unlike in the rest of the paper
knowledge of the
basic properties and techniques of the Ising model.

\begin{lem}\label{lem:notspacefilling}
For every positive $r$ there is a function $\delta_r(x):\,\RR_+\to\RR_+$ 
such that $\lim_{x\to0}\delta_r(x)=0$ and if an edge $e$
is $r$ away from at least one of the boundary arcs $ab$ or $ba$, then
$$\Prob{e\in\gamma}~\le~\delta_r(\mesh)~.$$
\end{lem}

\begin{proof}
Denote by $B$ and $W$ the centers of two neighboring squares (black and white),
separated by an edge $e$.
If an edge $e$ belongs to the interface, then
$B$ is connected by a cluster to the arc $ba$ and 
$W$ -- by a dual cluster to the arc $ab$ (since the interface separates the two).
Assume that the edge $e$
is at least $r$ away from the boundary arcs $ba$,
the case when it is away from $ab$ is treated similarly with clusters
replaced by dual clusters (which leads to the same result since the model is self-dual).
Denote by $Q$ the square with side length $r/2$ centered around $e$,
by our assumption it does not intersect the boundary arc $ba$.
Then by monotonicity we can write
\begin{align}
\PP&\br{e\in\gamma}~=\nonumber\\
=&\Prob{B\mathrm{~connected~by~a~cluster~to~the~arc~}ba
\mathrm{~and~}W\mathrm{~connected~by~a~dual~cluster~to~the~arc~}ab}\nonumber\\
\le&\Prob{B\mathrm{~connected~by~a~cluster~to~the~arc~}ba}\nonumber\\
\le&\Prob{B\mathrm{~connected~by~a~cluster~to~}\partial Q
\mathrm{~inside~}\Om
\mathrm{~with~}ba\mathrm{~wired~and~}ab\mathrm{~dual-wired}}\nonumber\\
\le&\Prob{B\mathrm{~connected~by~a~cluster~to~}\partial Q
\mathrm{~inside~}\Om
\mathrm{~with~}ba\cup ab\setminus Q\mathrm{~wired~and~}ab\cap Q\mathrm{~dual-wired}}\nonumber\\
=&\Prob{B\mathrm{~connected~by~a~cluster~to~}\partial Q
\mathrm{~inside~}\Om
\mathrm{~with~}\partial\Om\setminus Q\mathrm{~wired~and~}\partial\Om\cap Q\mathrm{~dual-wired}}\nonumber\\
\le&\Prob{B\mathrm{~connected~by~a~cluster~to~}\partial Q
\mathrm{~inside~}\Om
\mathrm{~with~}\Om\cap\partial Q\mathrm{~wired~and~}\partial\Om\cap Q\mathrm{~dual-wired}}\nonumber\\
\le&\Prob{B\mathrm{~connected~by~a~cluster~to~}\partial Q
\mathrm{~inside~}Q
\mathrm{~with~}\partial Q\mathrm{~wired~}}\nonumber\\
=&\mathrm{~magnetization~at~}B\mathrm{~in ~the~Ising~spin~model~in~}Q
\mathrm{~with~}``+"\mathrm{~boundary~conditions~on~}\partial Q\nonumber.
\nonumber
\end{align}
The right hand side, the
magnetization at criticality, 
was computed by Kaufman--Onsager and Yang \cite{\kaufmanonsageriv,\yang},
and for a fixed $r$ it tends to zero with the mesh $\mesh$ tending to zero,
proving the Lemma.

Note that magnetization was computed to tend to zero like $\mesh^{1/8}$, but we do
not need this stronger statement. On the other hand, after convergence to \slee~is 
established, one can obtain even better asymptotics for the magnetization,
deriving a conformally covariant constant in front of $\mesh^{1/8}$.
\end{proof}

\section{Estimates of discrete harmonic functions}\label{sec:dharm}

In this Section we collect the needed facts about discrete harmonic functions.
Let $\mesh_j\ZZ^2$ be a sequence of lattices with mesh $\mesh_j$ tending to zero.
For a domain $U$ we denote by $U_j$ the corresponding lattice domain
in  $\mesh_j\ZZ^2$.
As usual, functions on a lattice domain are thought of as defined in the whole domain,
say by piecewise constant or linear continuation.

A classical fact says that solution of the discrete Dirichlet problem 
converges uniformly away from the boundary
to the solution of the continuum one. 
The following Lemma can be found
in the seminal paper \cite{\courantetal}
by Courant, Friedrichs and Lewy.
It can also be deduced from the random walk interpretation
of harmonic functions (which is also discussed in \cite{\courantetal}).

\begin{lem}\label{lem:cfl}
Let $\Om_j$ be a sequence of lattice approximations to a smooth domain $\Om$
with the mesh $\mesh_j$ tending to zero.
Let $\brs{h_j}$ be a sequence of discrete harmonic functions on
lattice domains $\Om_j$
and $h$ be a harmonic function on $\Om$ with continuous boundary values.
If $h_j$ converge uniformly to $h$ on $\partial\Om$,
then inside $\Om$ away from the boundary
$h_j$ and its partial discrete derivatives (i.e. normalized differences) 
are equicontinuous and converge uniformly to 
their continuum counterparts, i.e.
$h$ and its partial derivatives.
\end{lem}

In our case the lattice domains approximate $\Om$
in the Carath\'eodory (rather than in the Hausdorff) sense,
furthermore the boundary values are discontinuous.
We will deduce a suitable result using the following
well-known estimate:

\begin{lem}\label{lem:beurling}
There exist an increasing positive function $\epsilon$ on $\RR_+$
with $\lim_{x\to0+}\epsilon(x)=0$, such that the following holds.
Let $H$ be a non-negative bounded discrete harmonic function in a 
simply-connected domain $\Om$
with boundary values equal to zero on $\Om\cap B(z,r)$
and at most one elsewhere.
If $\dist{z,\partial\Om}<\delta$,
then $H(z)<\epsilon(\delta/r)$.
\end{lem}

This is a weaker version of discrete Beurling's estimate
$\epsilon(\delta/r)=\const\sqrt{\delta/r}$.
It can be reformulated in terms of the hitting probabilities
for the simple random walk and is found
in Kesten's \cite{\kestenwalks}.

Now we can prove the needed version of the convergence Lemma:

\begin{lem}\label{lem:cfl1}
Suppose that as the lattice mesh $\mesh_j$ goes to zero,
the discrete domains $\Om_j$ (with points $a_j$, $b_j$ on the boundary)
converge to the domain $\Om$ (with points $a$, $b$ on the boundary)
in the Carath\'eodory sense.
Let $h_j$ be a discrete harmonic function on $\Om_j$
with boundary values on the arc $a_jb_j$ and $1$
on the arc $b_ja_j$.
Then inside $\Om$ discrete functions $h_j$ converge uniformly
to their continuum counterpart $h$,
which is harmonic in $\Om$
with  boundary values $0$ on the arc $ab$ and $1$
on the arc $ba$. 
\end{lem}

\begin{proof}
Being harmonic with bounded boundary values functions $h_j$ for large $j$ 
are equicontinuous inside $\Omega$ by \cite{\courantetal} or \cite{\verblunsky} 
-- see inequality(\ref{eq:verblunsky}) below.
Thus it is enough to show that any subsequential limit, say $h'$, coincides with $h$.

Fix small $R>0$ such that
two balls $B(a,R)$ and $B(b,R)$ are disjoint,
and denote their union by $W$.
Let $r<2R$ be the distance between the arcs $ab\setminus W$ and $ba\setminus W$.
Take $\delta<r/2$ and let $\Om^\delta$ be a subdomain of $\Om$ with smooth boundary 
which is $\delta/2$-close to the boundary of $\Om$.
Let $a^\delta$ and $b^\delta$ be two points on $\partial\Om^\delta$
which are $\delta/2$ close to $a$ and $b$ correspondingly.
Carath\'eodory convergence implies that for large enough $j$
subdomain $\Om^\delta$ is contained in $\Om_j$ and
its boundary $\partial\Om^\delta$ is contained in the $\delta$-neighborhood
of the boundary of $\Om_j$.
(The opposite inclusion might fail 
if $\Om_j$ contains long fjords of fixed diameter, which
however disappear in the Carath\'eodory limit if their width tends to zero).

Denote by $a^\delta b^\delta$ and $b^\delta a^\delta$ the counterclockwise
boundary arcs of $\partial\Om^\delta$.
By the Lemma~\ref{lem:beurling} for sufficiently large $j$ the function $h_j$
is at most $\epsilon(\delta/r)$ on $a^\delta b^\delta\setminus W$,
on the other hand having non-negative boundary values it is non-negative there:
$$0\le h_j\le\epsilon(\delta/r)~\mathrm{on}~a^\delta b^\delta\setminus W,$$
and similarly
$$1-\epsilon(\delta/r)\le h_j\le1~\mathrm{on}~b^\delta a^\delta\setminus W.$$
Being a subsequential limit, $h'$ also satisfies these inequalities.
Sending $\delta$ to zero, we deduce that
$h'$ has boundary values $0$ on $ab\setminus W$ and $1$ on $ba\setminus W$.

When $R$ goes to zero, so does $r<2R$, and we see that
$h'$ has boundary values $0$ on $ab$ and $1$ on $ba$,
and being bounded coincides with $h$.
\end{proof}

Let $Q$ be a square with side $\mesh L$ on the lattice $\mesh \ZZ^2$
and denote by  $9Q$ a nine times bigger square.
We will need the following continuity estimate
from Verblunsky's \cite{\verblunsky}:
if a function $h$ is discrete harmonic in a square $2Q$,
then on the square $Q$
\begin{equation}\label{eq:verblunsky}
\sup_{Q}\abs{\nabla h}\le \frac{\const}{L}\,\sup_{\partial9Q}\abs{h}.
\end{equation}

Let $G(x,y)=G_{9Q}(x,y)$ denote the discrete Green's function for the square $9Q$,
with $\Delta_x G(x,y)=\Delta_y G(x,y)=\delta_{x-y}$ and vanishing on the boundary of $9Q\times9Q$.
By $G_\CC$ we denote the discrete Green's function in the whole plane,
normalized so that $G_\CC(y,y)=0$. 
By the equation (9.6) in the paper \cite{\mccreawhipple} of McCrea and Whipple,
it satisfies 
\begin{equation}\label{eq:glog}
G_\CC(x,y)=\frac1\pi\log{\frac{\abs{x-y}}{\mesh}}+C+\oo\br{\frac{\mesh}{|x-y|}},~~~\frac{x-y}{\mesh}\to\infty,
\end{equation}
for a specific constant $C$ (which can be written in terms of the Euler constant).

We will need the following integral estimate of the gradient of $G$ in terms of $G$ itself:

\begin{lem}\label{lem:green}
There is a constant $\const$ independent of $L$ 
such that for every $y\in 9Q$ one has
\begin{equation}\label{eq:green}
\sum_{x\in Q}\abs{\nabla G(x,y)}<\frac{\const}{L}\sum_{x\in Q}\abs{G(x,y)}.
\end{equation}
\end{lem}

\begin{proof}
By adjusting the constant we can assume that $L$ is large enough.

Suppose first that $y\in2Q$.
Denote by $G_\CC^*(\cdot,y)$ the discrete harmonic function on $9Q$
having the same boundary values as $G_\CC(\cdot,y)$.
We note that
$$G(\cdot,y)=G_\CC(\cdot,y)-G_\CC^*(\cdot,y).$$
By (\ref{eq:glog}) on $\partial9Q$ we have
$$G_\CC^*(\cdot,y)=G_\CC(\cdot,y)>\frac1\pi\log\br{\frac{9-2}2L}+C+\oo\br{\frac1L},$$
and so by the maximum principle the same estimate holds 
for $G_\CC^*(\cdot,y)$ inside $9Q$.
On the other hand,  (\ref{eq:glog}) implies that
for $x\in Q$
$$G_\CC(x,y)<\frac1\pi\log\br{2\sqrt2 L}+C+\oo\br{\frac{\mesh}{\abs{x-y}}}.$$
Combining those inequalities we infer that
for $x\in Q$
$$G(x,y)<\frac1\pi\log\br{2\sqrt2 L}-\frac1\pi\log\br{3\frac12L}+\oo\br{\frac{\mesh}{\abs{x-y}}}
=-\frac1\pi\log{\frac3{2\sqrt2}}+\oo\br{\frac{\mesh}{\abs{x-y}}},$$
and summing over $Q$ (recall that $G$ is negative) we arrive at
\begin{equation}\label{eq:g}
\sum_{x\in Q}\abs{G(x,y)}\ge\const L^2.
\end{equation}

It follows from (\ref{eq:glog}) that $G_\CC^*(\cdot,y)$ is
equal on $\partial9Q$ to a constant function $\frac1\pi\log L+C$ up to an error term of 
$\frac1\pi\log\br{5\frac12\sqrt2}+\oo\br{\frac1L}$.
Therefore by (\ref{eq:verblunsky}) one has
\begin{equation}\label{eq:nablagstar}
\sum_{x\in Q}\abs{\nabla G_\CC^*(x,y)}\le\sum\frac1L \frac1\pi\log\br{5\frac12\sqrt2}+\oo\br{L}=\const L.
\end{equation}

Let $\ell$ be a lattice line through $y$, and $\ell'$ an orthogonal line intersecting
$\ell$ at $x$ and $\partial Q$ at $x'$ and $x''$.
By symmetry the whole plane Green's function $G_\CC(\cdot,y)$ is monotone
on the intervals $[x',x]$ and $[x,x'']$.
So using (\ref{eq:glog}) we estimate
the sum of absolute values of differences of $G$ along this line by
$$
G_\CC(x',y)+G_\CC(x'',y)-2G_\CC(x,y)\le2\log L -2\log\frac{\abs{x-y}}{\mesh}+\const.
$$
Summing this up for all lattice lines $\ell'$ in both directions, we arrive at
\begin{equation}\label{eq:nablagcc}
\sum_{x\in Q}\abs{\nabla G_\CC(x,y)}\le 8 \sum^L_{j=1}\br{\log L-\log j+\const}\le\const L.
\end{equation}

Combining(\ref{eq:g}), (\ref{eq:nablagstar}) and (\ref{eq:nablagcc}) 
we prove the Lemma in the case $y\in 2Q$:
$$\sum\abs{\nabla G}\le\sum\abs{\nabla G_\CC^*}+\sum\abs{\nabla G_\CC}
\le \const L \le \frac{\const}L \sum\abs{G}.$$

It remains to deal with the case $y\in9Q\setminus2Q$.
In this case $G(\cdot,y)$ is discrete harmonic and negative in $2Q$,
so its values on $Q$ are comparable by Harnack's principle
to its value at the center, say $A$.
Using (\ref{eq:verblunsky}) again, we write
$$
\sum\abs{\nabla G}\lesssim\sum\frac{A}{L}\asymp\frac1L \sum\abs{G},
$$
thus proving the Lemma.
\end{proof}

\section{Unpleasant computations for Lemma~\ref{lem:harm}}\label{sec:subharm}

There are several ways to prove equations (\ref{eq:superharm},\ref{eq:subharm})
and  we present not the shortest calculation,
but perhaps the most straight-forward one.

We will prove (\ref{eq:superharm}), the proof
of (\ref{eq:subharm}) is similar.
Let $u$ be the center of some white square.
Denote by $NW$, $NE$, $SE$, $SW$ its corner vertices,
starting from the upper left and going clockwise.
Recall that by Remark~\ref{rem:usualhol}
$$
F(NW)-F(SE)~=~i\,\br{F(NE)-F(SW)}~,
$$
so to prove (\ref{eq:subharm}) it is sufficient to show that
$$
\Delta\HH(W)~=~-\abs{F(NE)-F(SW)}^2~.
$$

To simplify calculations denote $\si:=\exp(-i\pi/8)$.
Denote by  $N$, $E$, $S$, $W$ the centers of bordering edges,
starting from the top and going clockwise.
Assume that the line $\ell(N)$ passes through a unit vector $\alpha$.
With the chosen orientation $\alpha=1$,
but we will leave $\alpha$ a variable 
to be able to compare results for different vertices.
Then the lines $\ell(W)$, $\ell(S)$, $\ell(E)$
pass through the vectors $\alpha\si^2$, $\alpha\si^4$, $\alpha\si^6$
correspondingly.
See Figure~\ref{fig:square}.

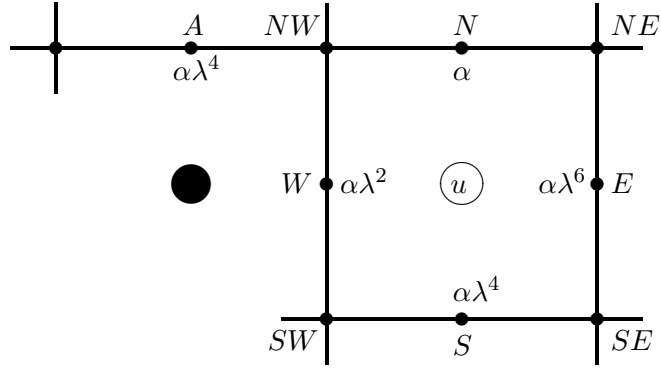
\begin{figure}
{\def\bs{\circle*{3}}
\unitlength=0.6mm
\begin{picture}(150,90)(-65,-15)
{\linethickness{1.pt}
\put(0,-10){\line(0,1){80}}\put(60,-10){\line(0,1){80}}
\put(-10,0){\line(1,0){80}}\put(-70,60){\line(1,0){140}}
\put(-60,50){\line(0,1){20}}}
\put(0,60)\bs\put(0,0)\bs\put(60,0)\bs\put(60,60)\bs\put(-60,60)\bs
\put(-13,-07){$SW$}\put(-14,63){$NW$}\put(63,63){$NE$}\put(63,-07){$SE$}
\put(-10,28){$W$}\put(63,28){$E$}\put(28,63){$N$}\put(28,-08){$S$}
\put(-32,63){$A$}
\put(3,28){$\alpha\si^2$}\put(47,28){$\alpha\si^6$}
\put(28,53){$\alpha$}\put(28,4){$\alpha\si^4$}
\put(-34,53){$\alpha\si^4$}
\put(0,30)\bs\put(30,0)\bs\put(60,30)\bs\put(30,60)\bs\put(-30,60)\bs
\put(-30,30){\circle*{10}}
\put(30,30){\circle{10}}
\put(27.5,28){$u$}
\end{picture}}
\caption{Vertices and edges around a square.
Lines corresponding to the edges
pass through the written vectors.}\label{fig:square}
\end{figure}

First evaluate increment $\partial_{NW}\HH$ of $\HH$ across the vertex $NW$.
Denote by $A$ the center of the edge going left from the vertex $NW$.
Recalling (\ref{eq:proj}), by definition of $\HH$ we write
\begin{align}
\partial_{NW}\HH&=\abs{F(W)}^2-\abs{F(A)}^2\nonumber\\
&=\abs{\Proj{F(NW),\alpha\si^2}}^2-\abs{\Proj{F(NW),\alpha\si^4}}^2\nonumber\\
&=\frac14\abs{F(NW)+\alpha^2\si^4\bar F(NW)}^2-\frac14\abs{F(NW)+\alpha^2\si^8\bar F(NW)}^2\nonumber\\
&=\frac14\br{F(NW)+\alpha^2\si^4\bar F(NW)}\br{\bar F(NW)+\bar\alpha^2\bar\si^4 F(NW)}\nonumber\\
&~~~-\frac14
\br{F(NW)+\alpha^2\si^8\bar F(NW)}\br{\bar F(NW)+\bar\alpha^2\bar\si^8 F(NW)}\nonumber\\
&=\frac14\br{2F(NW)\bar F(NW)+\alpha^2\si^4\bar F^2(NW)+\bar\alpha^2\bar\si^4 F^2(NW)}\nonumber\\
&~~~-\frac14
\br{2F(NW)\bar F(NW)+\alpha^2\si^8 \bar F^2(NW)+\bar\alpha^2\bar\si^8 F^2(NW)}\nonumber\\
&=\frac14
\br{\br{\si^2-\si^6}\alpha^2\si^2 \bar F^2(NW)
+\br{\bar\si^2-\bar\si^6}\bar\alpha^2\bar\si^2 F^2(NW)}\nonumber\\
&=\frac{\sqrt2}4
\br{\alpha^2\si^2 \bar F^2(NW)
+\bar\alpha^2\bar\si^2 F^2(NW)}\nonumber
\end{align}
Writing similarly increments across $SW$, $SE$, $NE$
(with $\alpha$ substituted by $\alpha\si^2$, $\alpha\si^4$, $\alpha\si^6$
correspondingly), we arrive at the following four equations:
\begin{align}
\partial_{NW}\HH&=\frac{\sqrt2}4
\br{\alpha^2\si^2 \bar F^2(NW)
+\bar\alpha^2\bar\si^2 F^2(NW)}\label{eq:nw}\\
\partial_{SW}\HH&=\frac{\sqrt2}4
\br{\alpha^2\si^6 \bar F^2(SW)
+\bar\alpha^2\bar\si^6 F^2(SW)}\label{eq:sw}\\
\partial_{SE}\HH&=\frac{\sqrt2}4
\br{\alpha^2\si^{10} \bar F^2(SE)
+\bar\alpha^2\bar\si^{10} F^2(SE)}\label{eq:se}\\
\partial_{NE}\HH&=\frac{\sqrt2}4
\br{\alpha^2\si^{14} \bar F^2(NE)
+\bar\alpha^2\bar\si^{14} F^2(NE)}\label{eq:ne}
\end{align}
Let us remark that from these equations it becomes
clear that $\HH$ is an appropriate 
discrete version of the primitive $\IM\int F^2 dz$.
Indeed, if $z$ is one of the corners and $v$ is the 
center of the square across that corner, 
denoting $\eta:=(v-u)/\abs{v-u}$ we see that
\begin{align}
(H(v)-H(u))&=
\frac{\sqrt2}4\br{i \bar \eta \bar F^2(z)-i \eta F^2(z)}
\nonumber\\&=
\frac{\sqrt2}2\,\IM\br{\eta F^2(z)}
\nonumber\\&=
\frac{1}{2\mesh}\,\IM\br{ F^2(z)\,(v-u)},\nonumber
\end{align}
and therefore
\begin{equation}\label{eq:hintf}
2\,\mesh\,(H(v)-H(u))=\IM \br{F(z)^2\,(v-u)}.
\end{equation}

Summing the equations (\ref{eq:nw},\ref{eq:sw},\ref{eq:se},\ref{eq:ne}), 
we can write the Laplacian $\Delta\HH(u)$
in terms of the values of $F$ at four neighboring vertices.
But we want to reduce this further to the values of $F$ at two vertices.
Such reduction is possible, since
by discrete analyticity projections of $F(NW)$ on lines
$\ell(N)$ and $\ell(W)$ coincide with those
of $F(NE)$ and $F(SW)$ correspondingly.
Using (\ref{eq:proj}) we can write that as
\begin{align}
F(NW)+\alpha^2\bar F(NW)&=F(NE)+\alpha^2\bar F(NE)~,\nonumber\\
F(NW)+\alpha^2\si^4\bar F(NW)&=F(SW)+\alpha^2\si^4\bar F(SW)~.\nonumber
\end{align}
Subtracting the equations multiplied  by $\si^2$ and $\bar\si^2$
correspondingly we arrive at
$$
(\si^2-\bar\si^2)F(NW)=\si^2F(NE)+\alpha^2\si^2\bar F(NE)-\bar\si^2F(SW)-\alpha^2\si^2\bar F(SW)~,
$$
where $\si^2-\bar\si^2$ simplifies to $-i\sqrt2$:
\begin{equation}
F(NW)=\frac{i}{\sqrt2}\br{\si^2F(NE)+\alpha^2\si^2\bar F(NE)-\bar\si^2F(SW)-\alpha^2\si^2\bar F(SW)}~.\label{eq:renw}
\end{equation}
Writing similarly for $F(SE)$ (with $\alpha\si^4$
substituted for $\alpha$ and $NE$ and $SW$ interchanged) we conclude that
\begin{equation}
F(SE)=\frac{i}{\sqrt2}\br{\si^2F(SW)+\alpha^2\si^{10}\bar F(SW)-\bar\si^2F(NE)-\alpha^2\si^{10}\bar F(NE)}~,
\label{eq:rese}\end{equation}

Now we can sum equations (\ref{eq:nw},\ref{eq:se},\ref{eq:sw},\ref{eq:ne}),
substituting (\ref{eq:renw},\ref{eq:rese}) for values of $F(NW)$ and $F(SE)$:
\begin{align}
\Delta\HH(w)&=\partial_{NW}\HH+\partial_{SW}\HH+\partial_{SE}\HH+\partial_{NE}\HH\nonumber\\
&= 
\frac{\sqrt2}4\Bigg(
\alpha^2\si^2 \br{\,\overline{
\frac{i}{\sqrt2}\br{\si^2F(NE)+\alpha^2\si^2\bar F(NE)-\bar\si^2F(SW)-\alpha^2\si^2\bar F(SW)}}\,}^2
\nonumber\\
&\qquad+\bar\alpha^2\bar\si^2 
\br{\frac{i}{\sqrt2}\br{\si^2F(NE)+\alpha^2\si^2\bar F(NE)-\bar\si^2F(SW)-\alpha^2\si^2\bar F(SW)}}^2
\nonumber\\
&\qquad+
\alpha^2\si^{10} 
\br{\,\overline{\frac{i}{\sqrt2}\br{\si^2F(SW)+\alpha^2\si^{10}\bar F(SW)-\bar\si^2F(NE)-\alpha^2\si^{10}\bar F(NE)}}\,}^2
\nonumber\\
&\qquad+\bar\alpha^2\bar\si^{10} 
\br{\frac{i}{\sqrt2}\br{\si^2F(SW)+\alpha^2\si^{10}\bar F(SW)-\bar\si^2F(NE)-\alpha^2\si^{10}\bar F(NE)}}^2
\nonumber\\
&\qquad+\alpha^2\si^6 \bar F^2(SW)
+\bar\alpha^2\bar\si^6 F^2(SW)
+\alpha^2\si^{14} \bar F^2(NE)
+\bar\alpha^2\bar\si^{14} F^2(NE)\Bigg)
\nonumber
\end{align}
When we plug in $\alpha=1$ and recall that $\si=\exp(-i\pi/8)$,
in particular $\si^8=-1$,
there are many cancelations in the right hand side,
which eventually simplifies:
\begin{align}
\dots&=\frac{\sqrt2}4\bigg(
\si^2 \br{{
\,-\frac{i}{\sqrt2}\br{\bar\si^2\bar F(NE)+\bar \si^2 F(NE)-\si^2\bar F(SW)-\bar \si^2 F(SW)}}\,}^2
\nonumber\\
&\qquad+\bar\si^2 
\br{\frac{i}{\sqrt2}\br{\si^2F(NE)+\si^2\bar F(NE)-\bar\si^2F(SW)-\si^2\bar F(SW)}}^2
\nonumber\\
&\qquad-
\si^{2} 
\br{\,{-\frac{i}{\sqrt2}\br{\bar \si^2\bar F(SW)-\bar \si^{2}F(SW)-\si^2\bar F(NE)+\bar \si^{2} F(NE)}}\,}^2
\nonumber\\
&\qquad-\bar\si^{2} 
\br{\frac{i}{\sqrt2}\br{\si^2F(SW)-\si^{2}\bar F(SW)-\bar\si^2F(NE)+\si^{2}\bar F(NE)}}^2
\nonumber\\
&\qquad+\si^6 \bar F^2(SW)
+\bar\si^6 F^2(SW)
-\si^{6} \bar F^2(NE)
-\bar\si^{6} F^2(NE)\bigg)
\nonumber\\
&=\frac{\sqrt2}4\bigg(
-\frac12\Bbr{\bar\si\bar F(NE)+\bar \si F(NE)-\si^3\bar F(SW)-\bar \si F(SW)}^2
\nonumber\\
&\qquad-\frac12 
\Bbr{\si F(NE)+\si\bar F(NE)-\bar\si^3F(SW)-\si\bar F(SW)}^2
\nonumber\\
&\qquad+\frac12
\Bbr{\bar \si\bar F(SW)-\bar \si F(SW)-\si^3\bar F(NE)+\bar \si F(NE)}^2
\nonumber\\
&\qquad+\frac12 
\Bbr{\si F(SW)-\si\bar F(SW)-\bar\si^3 F(NE)+\si \bar F(NE)}^2
\nonumber\\
&\qquad+\si^6 \bar F^2(SW)
+\bar\si^6 F^2(SW)
-\si^{6} \bar F^2(NE)
-\bar\si^{6} F^2(NE)\bigg)
\nonumber\\
&=\frac{\sqrt2}8\bigg(
F^2(NE)\br{-\bar\si^2-\si^2+\bar\si^2+\bar\si^6-2\bar\si^6}+
\bar F^2(NE)\br{-\bar\si^2-\si^2+\si^6+\si^2-2\si^6}\nonumber\\
&\qquad+F^2(SW)\br{-\bar\si^2-\bar\si^6+\bar\si^2+\si^2+2\bar\si^6}+
\bar F^2(SW)\br{-\si^6-\si^2+\bar\si^2+\si^2+2\si^6}\nonumber\\
&\qquad+2F(NE)F(SW)\br{\bar\si^2+\bar\si^2-\bar\si^2-\bar\si^2}
+2\bar F(NE)\bar F(SW)\br{\si^2+\si^2-\si^2-\si^2}\nonumber\\
&\qquad+2F(NE)\bar F(NE)\br{-\bar\si^2-\si^2-\si^2-\bar\si^2}
+2F(SW)\bar F(SW)\br{-\si^2-\bar\si^2-\bar\si^2-\si^2}\nonumber\\
&\qquad+2F(NE)\bar F(SW)\br{\si^2+\si^2+\bar\si^2+\bar\si^2}
+2F(SW)\bar F(NE)\br{\bar\si^2+\bar\si^2+\si^2+\si^2}\bigg)\nonumber\\
&=-\frac{\sqrt2}84\sqrt2\br{F(NE)-F(SW)}\br{\bar F(NE)-\bar F(SW)}
=-\abs{F(NE)-F(SW)}^2~.
\nonumber
\end{align}
This finishes the proof of Lemma.

\bibliographystyle{plain}

\end{document}